\documentclass[a4paper,12pt]{article}
\usepackage{physics}
\usepackage{color}
\usepackage{amsmath, amssymb, amsfonts}
\usepackage{geometry}
\usepackage{graphicx}
\usepackage{titlesec}
\usepackage{relsize}
\usepackage{bbold} 
\graphicspath{ {images/} }
\numberwithin{equation}{section} 
\usepackage{amsthm}
\usepackage{dsfont}
\newtheorem{theorem}{Theorem}[section]
\newtheorem{lemma}[theorem]{Lemma}

\newtheorem{corollary}[theorem]{Corollary}
\newtheorem{definition}[theorem]{Definition}

\newtheorem{assumption}[theorem]{Assumption}
\usepackage{mwe}
\usepackage{epsfig}
\usepackage{hyperref}
\usepackage{cleveref}
\usepackage{algorithm}
\usepackage[noend]{algpseudocode}
\usepackage{multirow}
\usepackage{makecell}
\usepackage{subcaption}
\usepackage{yfonts}
\usepackage{tikz}
\usepackage{amsmath,mathrsfs}
\usepackage{relsize}
\usepackage{tikz}
\usetikzlibrary{tikzmark,calc}

\usepackage{geometry,mathtools}
\usetikzlibrary{calc}
\usetikzlibrary{matrix}
\usetikzlibrary{decorations.pathreplacing}

\DeclareMathOperator{\supp}{supp}

\DeclareMathOperator{\poly}{poly}
\newcommand{\todo}[1]{\textcolor{red}{#1}}

\begin{document}
	
	\title{Detailed Analysis of Circuit-to-Hamiltonian Mappings}
	\author{James D. Watson}
	\date{Department of Computer Science, University College London, UK}
	\maketitle

\begin{abstract}	
	The circuit-to-Hamiltonian construction has found widespread use within the field of Hamiltonian complexity, particularly for proving QMA-hardness results. In this work we examine the ground state energies of the Hamiltonian for standard clock constructions and those which require dynamic initialisation. 
	We put exponentially tight bounds on these ground state energies and also determine improved scaling bounds in the case where there is a constant probability of the computation being rejected.
	Furthermore, we prove a collection of results concerning the low-energy subspace of quantum walks on a line with energy penalties appearing at any point along the walk and introduce some general tools that may be useful for such analyses.
\end{abstract}	
	
\tableofcontents	
	
\section{Introduction}

In the previous two decades there has been a union between condensed matter physics and complexity theory resulting in the new field of Hamiltonian complexity. The idea of Hamiltonian complexity is to study the properties of local Hamiltonians from a complexity perspective, allowing us understand many-body quantum systems that may be to complicated to solve or are otherwise computationally intractable.
Key to this field is the Feynman-Kitaev Hamiltonian construction which allows the evolution of a quantum circuit to be encoded in a many-body Hamiltonian and was used in the seminal proof that the task of estimating the ground state of the local Hamiltonian problem is QMA-complete \cite{Kitaev_Shen_Vyalyi}. 
This technique is often called a circuit-to-Hamiltonian mapping. 
Initial work used the mapping to prove QMA-hardness of the Local Hamiltonian problem for 5-local Hamiltonians \cite{Kitaev_Shen_Vyalyi} before it was then used to prove completeness of progressively more restricted classes of Hamiltonians. 
The culmination of the circuit to Hamiltonian mapping has been proving QMA-hardness of 1D Hamiltonians \cite{aharonov_gottesman_irani_kempe_2009}, and even 1D, translationally invariant, nearest neighbour Hamiltonians \cite{Gottesman-Irani}. 

Perturbation gadgets developed first by \cite{Kempe_Kitaev_Regev}, combined with the circuit-to-Hamiltonian mapping, have been used to prove further sets of hardness results for systems on a lattice \cite{Oliviera_Terhal}, and classify 2-local qubit Hamiltonians \cite{Cubitt_Montanaro_2013}. 
Work has also been done to classify the complexity of sampling \cite{Aaronson_Arkhipov_2011}, traversing the ground state subspace \cite{Gharibian_Sikora_2015}, counting low energy states \cite{Brown_Flammia_Schuch_2011}, excited states \cite{Jordan_Gosset_Love_2010}, and finding the spectral gap \cite{Ambainis_2014, Gharibia_Yirka_2016}. 
The circuit-to-Hamiltonian construction has found usage in wide range of other areas in quantum information, including adiabatic quantum computation \cite{aharonov_dam_kempe_landau_lloyd_regev_2008} and error correction \cite{bohdanowicz_crosson_nirkhe_yuen_2019}.

Detailed analysis has been able to given improved bounds on the promise gap and eigenvalue scaling, not just for the standard constructions but also for non-uniform weights and for branching computations \cite{Caha_Landau_Nagaj, Bausch_Crosson_2018, Usher_Hoban_Browne_2017, Bausch_Cubitt_Ozols_2016}.
However, the analysis of the gaps and eigenvalues are largely just scaling analysis. The aim of this paper is to pin down as closely as possible the eigenvalues and ground states of the Feynman-Kitaev construction. 

The contribution from this work is extending the analysis to constructions with more complex clock constructions which have been used in prior literature (in particular in the 1D case \cite{Cubitt_Perez-Garcia_Wolf, Gottesman-Irani}) which include clocks which require dynamic initialisation and do not have penalty terms only at the beginning and end. 
Furthermore, we pin down the ground state energy of these circuit-to-Hamiltonian mappings to within exponential precision. 

We also prove some potentially useful minimum eigenvalue bounds for Hamiltonians of quantum walks on a line with penalties (or self-loops) which improve on bounds by \cite{Caha_Landau_Nagaj} and use different methods. 
Using these, we prove tighter bounds for Hamiltonians which encode computations which reject with constant probability.

\section{Background and Previous Results}

In this section we give a brief overview of the Feynman-Kitaev circuit-to-Hamiltonian mapping and some previously known results about its properties.

\subsection{Preliminaries}

\begin{definition}[QMA]\label{Def:QMA}
	A promise problem $A=(A_{YES},A_{NO})$ is in QMA if there exist polynomials $p$ and $q$ and a QTM $M$ such
	that for each instance $x$ and any quantum witness $\ket{w}$ such that $\ket{w}$ is of at most $q(|x|)$ qubits $M$ halts in $p(|x|)$ steps on input $(x, \ket{w})$, and
	\begin{itemize}
		\item if $x \in A_{YES}$,  $\exists \ket{w}$ such that $M$ accepts $(x, \ket{w})$ with probability $>2/3$. 
		\item if $x \in A_{NO}$ then $\forall \ket{w}$, $M$ accepts $(x, \ket{w})$ with probability $<1/3$.
	\end{itemize}
\end{definition}

As an intermediate step we will find it useful to prove results about EQMA (Exact-QMA) which is a zero-error quantum complexity class:
\begin{definition}[EQMA] \label{Def:EQMA}
	A promise problem $A=(A_{YES},A_{NO})$ is in EQMA if there exist polynomials $p$ and $q$ and a QTM $M$ such
	that for each instance $x$ and any quantum witness $\ket{w}$ such that $\ket{w}$ is of at most $q(|x|)$ qubits $M$ halts in $p(|x|)$ steps on input $(x, \ket{w})$ and
	\begin{itemize}
		\item if $x \in A_{YES}$,  $\exists \ket{w}$ such that $M$ accepts $(x, \ket{w})$ with probability 1. 
		\item if $x \in A_{NO}$ then $\forall \ket{w}$, $M$ accepts $(x, \ket{w})$ with probability 0 (i.e. always rejects).
	\end{itemize}
\end{definition}
We note that EQMA is not a particularly ``natural'' class and suffers the same ambiguities that EQP does (as defined in \cite{Bernstein_Vazirani_1997}). 
However, throughout the next few sections we will find it is easier prove results for the problem class EQMA as an intermediate step before using the EQMA results to prove results about QMA. 
We also take care to distinguish EQMA from the class NQP, defined in \cite{Adleman_Demarrais_Huang_1997}, which has zero amplitude on the accept state if it is a rejecting instance, but is only required to have non-zero amplitude on the accept state when it is an accepting instance. 
It may still have non-zero amplitude on the reject state when the instance is an accepting instance whereas EQMA does not. 

Throughout the rest of the this work we will denote the matrix representing an $N$ vertex path graph Laplacian as
\begin{equation}\label{eq:E_matrix}
\Delta^{(N)} = \begin{pmatrix}
\frac{1}{2}  & -\frac{1}{2} &      0       &    \dots     &        &    \dots     &       0      &       0      \\
-\frac{1}{2} &       1      & -\frac{1}{2} &    \ddots    &        &              &              &       0      \\
0      & -\frac{1}{2} &      1       & -\frac{1}{2} &        &              &              &    \vdots    \\
\vdots    &    \ddots    & -\frac{1}{2} &      1       & \ddots &              &              &              \\
&              &              &    \ddots    & \ddots &              &    \ddots    &    \vdots    \\
\vdots    &              &              &              &        &    \ddots    &    \ddots    &       0      \\
0      &              &              &              & \ddots &    \ddots    &       1      & -\frac{1}{2} \\
0      &       0      &    \dots     &              & \dots  &      0       & -\frac{1}{2} & \frac{1}{2}
\end{pmatrix}_{N\times N}.
\end{equation}
Furthermore, we denote an $N\times N$ matrix of zeros as $0_{N}$.

\subsection{Feynman's Construction}
Consider a quantum circuit described by the unitary $U=U_N\dots U_2U_1$ (or alternatively a quantum TM evolving according to this set of unitaries). 
Further consider a set of qudits such that the total Hilbert space $\mathcal{H} = \mathcal{H}_{clock}\otimes \mathcal{H}_{register}$. 
$\mathcal{H}_{clock}$ will contain a set of clock states $\ket{0}, \ket{1}, \dots \ket{T-1}$, which will label the steps of the circuit after each unitary.
$\mathcal{H}_{register}$ will be the computational register state that is acted on by the unitaries.

We then design a Hamiltonian which has a ``history state'' ground state of the form
\begin{align}
\ket{\Psi} = \frac{1}{\sqrt{T}}\sum_{t=0}^{T-1}\ket{t}\otimes (U_t\dots U_2U_1)\ket{\phi},
\end{align} 
where $\ket{\phi}$ is the initial input state to the circuit.
To do this we choose 
\begin{align}
H = H_{trans} + H_{in},
\end{align}
where $H_{trans}$ encodes the transitions/propagation of the circuit as
\begin{align}
H_{trans} = \sum_{t=0}^{T-1}(\ket{t+1}\otimes U_t - \ket{t})(\bra{t+1}\otimes U_t - \bra{t}),
\end{align}
and where the term $H_{in}:=\ket{0}\bra{0}_{clock}\otimes\ket{00\dots0}\bra{00\dots0}_{ancillas}$ applies a penalty if the ancilla qubits in the computational register do not start in the all zeros state.
For a fixed local Hilbert dimension, the clock states will generally be $\log(N)$-local. 
However, typically there are ways of reducing the locality of the interaction to a constant \cite{Kitaev_Shen_Vyalyi}.

Often the circuit-to-Hamiltonian mapping is used with the aim of encoding a verifications circuit for a QMA problem. 
With this in mind one includes an output penalty 
\begin{align}\label{Eq:FK_Hamiltonian}
H_{FK} := H_{in} + H_{out} + H_{trans}
\end{align}
where $H_{out} = \ket{T}\bra{T}\otimes\Pi_{out}$, where $\Pi_{out}$ is a projector onto a rejection flag output by the quantum circuit (i.e. it penalises rejecting computations). 

\subsection{Clock Constructions}

The clock construction encoded in the Hamiltonian needs to be local and otherwise satisfy the constraints of the Hamiltonian. 
There are a multitude of clock constructions in the literature, notably including delocalised clocks \cite{Breuckmann_Terhal_2014} and translationally invariant clock constructions \cite{Gottesman-Irani, Cubitt_Perez-Garcia_Wolf}.
Due to the constraints of encoding a clock construction into a translationally invariant Hamiltonian, these clocks require a dynamic initialisation and may undergo ``bad'' transitions (transitions to states that should not be allowed but cannot otherwise be excluded).
We include them in our analysis below. 
Previous analyses of the circuit-to-Hamiltonian mapping have
only achieved loose scaling bounds for such clock constructions, which we improve on here.

\subsection{Spectra of Feynman-Kitaev Hamiltoninans and the Promise Gap}

The Feynman-Kitaev Hamiltonian is often invoked to prove QMA-hardness results. 
Here, we use the fact that if the Hamiltonian encodes a rejecting instance, it will have a high energy ground state, $\lambda_0>\beta$, otherwise if it encodes an accepting instance, it will have a low energy ground state $\lambda_0<\alpha$, for $\beta-\alpha = O(1/\poly(n))$ for a Hamiltonian on $n$ qudits. 
This separation in $\alpha, \beta$ is known as the promise gap.

In the original proof of QMA-hardness of the Local Hamiltonian problem, Kitaev's geometrical lemma is used to prove that the promise gap scales as $\Omega(T^{-3})$.
Both \cite{Bausch_Crosson_2018} and \cite{Caha_Landau_Nagaj} improve on this to show that the standard Feynman-Kitaev Hamiltonian has a promise gap which scales a $\Omega(T^{-2})$.
\cite{Bausch_Crosson_2018} also bounds the scaling for many non-uniform types of Hamiltonians which we will not be concerned with in this work.  

A further reason for interest in the promise gap comes from its relation to the quantum PCP conjecture \cite{Aharonov_Arad_Vidick_2013} -- if it were possible to produce a sufficiently large promise gap with a Feynman-Kitaev Hamiltonian then the PCP conjecture would follow. 

One can also consider alternative models of computation which branch, such as Quantum Thue Systems \cite{Bausch_Cubitt_Ozols_2016}. 
It has been shown, using a generalisation of Kitaev's geometrical lemma, that these models also have a promise gap $\Omega(N^{-3})$ where $N$ is the number of vertices in the unitary labelled graph representing the computation (Lemma 44 of \cite{Bausch_Cubitt_Ozols_2016}). 
Further work \cite{Bausch_Crosson_2018} shows that such constructions cannot be straightforwardly used to prove the quantum PCP conjecture.

Finally, it is worth noting that all known history state constructions in the literature have been shown to have a spectral gap that closes as the length of the computation they encode increases \cite{Cubitt_Gonzalez-Guillen_2018}. A similar result was shown in \cite{Crosson_Bowen_2017}.

In this work we pin down the promise gap for Feynman-Kitaev Hamiltonians with uniform weight transition rules to be within exponential precision of a fixed function $1-\cos(\pi/2T)$, where $T$ is the runtime of the computation.

\subsection{Quantum Walks on a Line}


It is well known that the Hamiltonian describing a particle hopping along a line of length $T$, where individual states are given by $\{\ket{1}, \ket{2}, \dots \ket{T}, \}$, is given by 
\begin{align}
H^{walk} &= \sum_{t=1}^{T-1} (\ket{t+1}-\ket{t})(\bra{t+1}-\bra{t}) \\
&= \Delta^{(T)}. 
\end{align}
This can be used to represent not just a particle propagating along a line, but a generic quantum process evolving. 
Throughout the rest of the paper, we will use a ``walk'' to refer to a graph Laplacian with weighted vertices. 

Our interest will be when the propagating process is a computation and the states represent the clock register labelling stages of the computation. 
In particular, the analysis of Feynman-Kitaev Hamiltonians, such as in equation \ref{Eq:FK_Hamiltonian}, can often be mapped to quantum walks on a line. 

In the event that a computation gets an energy penalty it is often possible to show that the analysis becomes equivalent to analysing a Laplacian plus projectors for the relevant time steps as below
\begin{align}
W^\dagger H^{walk}W = \Delta^{(T)}+\sum_{k\in K} \ket{k}\bra{k},
\end{align}
for some $K\subseteq \{1,2,\dots,T\}$.
 
With this in mind, \cite{Caha_Landau_Nagaj} put bounds on the scaling of the ground state energy for quantum walks on a line with penalties at their end points. 
In this work we improve on these bounds and introduce a set of techniques useful for analysing quantum walks on a line that have an energy penalty at a point which is not necessarily at the ends of the line.

\section{Main Results}

\subsection{Feynman-Kitaev Hamiltonians}

We consider an extension of Feynman-Kitaev Hamiltonians to a more general class of Hamiltonians that have been considered before which we call \emph{Standard Form Hamiltonians} that includes those which have ``bad'' clock transitions and clocks which may require a dynamic initialisation. 
Such clocks have appeared in \cite{Gottesman-Irani}, \cite{Cubitt_Perez-Garcia_Wolf}, which are notable for encoding QTM rather than circuits. 
We then consider a QMA verification computation encoded in the Hamiltonian and the associated minimum eigenvalues in both the accept and reject instances.

We show that the ground state energy of a Hamiltonian encoding a verification of a QMA YES or NO instance is given by the following theorem.
\begin{theorem}
	The ground state energy of a standard-form Hamiltonian, $H_{QMA}\in \mathcal{B}(\mathds{C}^d)^{\otimes n}$, encoding the verification computation of a QMA instance with total runtime $T=\poly(n)$ is bounded as
	\begin{align}
	0 \leq &\lambda_0\big( H_{QMA}^{(YES)}\big) \leq e^{-O(\poly(n))} \\
	1-\cos\bigg(\frac{\pi}{2T}\bigg) - e^{-O(\poly(n))} \leq  &\lambda_0\big( H_{QMA}^{(NO)}\big) \leq 1-\cos\bigg(\frac{\pi}{2T}\bigg).
	\end{align}
\end{theorem}
Although these bounds do not give a better promise gap scaling compared to other known results --- it remains $\Theta (T^{-2})$ as per \cite{Bausch_Crosson_2018} and \cite{Caha_Landau_Nagaj} --- the above gives the minimum eigenvalue more precisely.
Moreover, it gives bounds in the case where the clock has inherent bad transitions and so the Hamiltonian needs penalising terms and where the clock requires a dynamic initialisation --- cases not covered by \cite{Caha_Landau_Nagaj} or \cite{Bausch_Crosson_2018}.

Furthermore, we consider the case where the QMA acceptance probability is a constant and amplification is no possible --- for example for the class StoqMA.
\begin{theorem}[YES Instance Upper Bound] 
	Let $H^{(YES)}_{QMA}$ encode the verification of a YES QMA instance. Let $\eta=O(1)$ be the maximum probability of rejection, then 
	\begin{align}
	0\leq \lambda_0\big( H_{QMA}^{(YES)}\big) =O\bigg(\frac{\eta}{T^2}\bigg).
	\end{align}
\end{theorem}
\noindent 
This is an improvement on the bound found in \cite{Caha_Landau_Nagaj} of $O(\eta/T)$ and is expected to be tight.

\subsection{Quantum Walks on a Line}

We present new eigenvalue bounds on Hamiltonians which encode quantum walks on a line where a penalty is applied, giving particular consideration to the case where the penalty is at the end of the walk. 
This analysis will later be applied to Hamiltonians encoding computation which is incorrect in some way.

We also introduce the Uncoupling Lemma which allows the ground state energy of a quantum walk with a penalty on it to be analysed as two disjoint walks that share the penalty between them.
We state and prove this in the next section.

From this we are able to prove that if there is a walk on a line with a penalty term somewhere, then the lowest ground state energy is achieved with the penalty at one of the ends
\begin{lemma}
	Consider the Hamiltonian   
	\begin{equation}
	\Delta^{(T)} +  \ket{k}\bra{k}
	\end{equation}
	for some basis state $\ket{k}$ with $1\leq k \leq T$. Then 
	\begin{equation}
	\min_k\bigg( \lambda_0 (\Delta^{(T)} +  \ket{k}\bra{k})  \bigg) = 1- \cos\bigg( \frac{\pi}{2T}\bigg)
	\end{equation}
	which occurs for $k=1,T$ for some $T> T_0$.
\end{lemma}

As an extension of this, we place bounds on the minimum energy eigenvalue for quantum walks with weight $\mu\geq 0$ penalties: $H = \Delta^{(T)}+\mu\ket{T}\bra{T}$ and improve on the scaling in terms of $\mu$:
\begin{theorem}
	For $H = \Delta^{(T)} + \mu\ket{T}\bra{T}$ and $\mu=k/T$, then for sufficiently large $T$ we have 
	\begin{align}
	\lambda_0(H) &= \Theta\bigg( \frac{k}{T^2}\bigg)
	\end{align}
	where $k=O(1)$ is some constant.
\end{theorem}

\section{Hamiltonian Analysis for Quantum Walks on a Line}

Before we start further, we introduce a simple lemma that may find use elsewhere. 
Given a quantum walk on a line that receives an energy penalty 1 on the $k^{th}$ step along its propagation, then the ground state energy can be bounded from below by decoupling the Hamiltonian into two disjoint quantum walk Hamiltonians and sharing the energy penalty between the two new walks. 
\begin{lemma}[Uncoupling Lemma] \label{Lemma:Divorce_Lemma}
 	Given a matrix 
	\begin{align}
	H = \Delta^{(T)}+ \ket{k}\bra{k},
	\end{align}
	then 
	\begin{align}
	H \geq (\Delta^{(k-1)}+\frac{1}{4}\ket{k-1}\bra{k-1}) \oplus (\Delta^{(T-k+1)} + \frac{1}{2}\ket{k}\bra{k}).
	\end{align}
\end{lemma}
\begin{proof}
	We see that we can make a simple decomposition:
	\begin{align}
	H =& \Delta^{(T)}+ \ket{k}\bra{k} \\
	=&  (\Delta^{(k-1)} + \frac{1}{4}\ket{k-1}\bra{k-1})\oplus 0_{T-k+1} + J + (0_{k-1}\oplus\Delta^{(T-k+1)})
	\end{align}
	where
	\begin{align}
	J = 0_{k-1} \oplus \begin{pmatrix}
	&1/4 &-1/2 \\
	&-1/2 & 1
	\end{pmatrix}
	\oplus 0_{T-k}.
	\end{align}
	We note that $J\geq 0$, hence 
	\begin{align}
	H &\geq   (\Delta^{(k-1)} + \frac{1}{4}\ket{k-1}\bra{k-1})\oplus 0_{T-k+1} + (0_{k-1}\oplus\Delta^{(T-k+1)}  + \frac{1}{2}\ket{k}\bra{k}).
	\end{align}
\end{proof}
\noindent
Using this we can directly bound eigenvalues of $H$.

\subsection{Laplacian Matrix Analysis}

Before we begin we gather some results about tridiagonal toeplitz matrices.
Let 
\begin{equation}\label{eq:Tridiagonal_Toeplitz_Matrix}
A^{(n)} = \begin{pmatrix}
b+\gamma  & c &      0       &    \dots     &        &    \dots     &       0      &       \alpha      \\
a &       b      & c &    \ddots    &        &              &              &       0      \\
0      & a &      b       & c &        &              &              &    \vdots    \\
\vdots    &    \ddots    & a &      b       & \ddots &              &              &              \\
&              &              &    \ddots    & \ddots &              &    \ddots    &    \vdots    \\
\vdots    &              &              &              &        &    \ddots    &    \ddots   &       0      \\
0      &              &              &              & \ddots &    \ddots    &   b      & c \\
\beta      &       0      &    \dots     &          & \dots  &      0     & a & b+ \delta
\end{pmatrix}_{n\times n}.
\end{equation}
then the following is true for specific instances of this family of matrices.
\begin{lemma}[Theorem 3.4 \textit{(iv)} of \cite{Yueh_Cheng_2008}] \label{Lemma:Tridiagonal_Full_Penalty}
	Suppose that $a = c \neq 0$, $\gamma \delta - \alpha \beta = - a^2$, and $\alpha + \beta = \gamma + \delta = 0$. Define
	\begin{equation}
	\theta_k = \frac{(2k-1)\pi}{2n}, \quad k=1,2,\dots,n.
	\end{equation} 
	Then the eigenvalues of $A^{(n)}$ are given by
	$\lambda_k = b + 2a \cos \theta_k $.
\end{lemma}

\begin{lemma}[Theorem 3.2 \textit{(viii)} of \cite{Yueh_Cheng_2008}] \label{Lemma:Tridiagonal_Partial_Penalty}
	Suppose that $a = c \neq 0$, $\alpha \beta  = \gamma \delta$, $\alpha + \beta = 0$ and $\gamma + \delta = a$. Define
	\begin{equation}
	\theta_k = \frac{(2k-1)\pi}{2n+1}, \quad k=1,2,\dots,n.
	\end{equation}
	Then the eigenvalues of $A^{(n)}$ are given by
	$\lambda_k = b + 2a \cos \theta_k$.
\end{lemma}
\noindent
We note that $\Delta^{(n)}$ is a special case of $A^{(n)}$ with $a=c=-1/2$, $b=1$, $\gamma=\delta=-1/2$ and $\alpha=\beta=0$. 
Hence the above lemmas will be useful in proving the following:
\begin{lemma}[Starting Penalty Lemma] \label{Lemma:Minimum_Eigenvalue}
	Consider the $T\times T$ matrix  
	\begin{equation}
	\Delta^{(T)} +  \ket{k}\bra{k}
	\end{equation}
	for some basis state $\ket{k}$. Then 
	\begin{equation}
	\min_k\bigg( \lambda_0 \big(\Delta^{(T)} +  \ket{k}\bra{k}\big)  \bigg) = 1- \cos\bigg( \frac{\pi}{2T}\bigg)
	\end{equation}
	which occurs for $k=1,T$ for some $T> T_0$.
\end{lemma}
\begin{proof} 
	For the $k=1$ case, Lemma \ref{Lemma:Tridiagonal_Full_Penalty} gives the minimum eigenvalue of the above as 
	\begin{equation} 
	\lambda_0^{(k=1)} = 1-\cos(\frac{\pi}{2T}). 
	\end{equation}
	We now need to consider the $k=2,3$ cases separately. 
	\\
	\\
	\noindent
	$\mathbf{k=2:}$ \newline
	\\
	For the $k=2$ we consider splitting the matrix up into two separate matrices in the following way:
	\begin{scriptsize}
		
	\begin{align} \label{Eq:k=2_Case}
	\begin{pmatrix}
	1/2 &       -1/2&    &        &      & \\
	-1/2& \mathbf{2}&-1/2&        &      &  \\
	    &-1/2       &  1 &        &      &    \\
	    &           &    &  \ddots&  -1/2&       \\
	    &           &    &  -1/2  &  1   & -1/2  \\
	    &           &    &        &-1/2  & 1/2     \\
	\end{pmatrix} &=
	\begin{pmatrix}
	1/4& -1/2&&&& \\
	-1/2& 1&&&&  \\
	&&  &   &  &    \\
	&&&  \mathbf{0}&  &   &    \\
	&&&      &    &   \\
	&&&              & &  \\
	\end{pmatrix} \\
	&+ \begin{pmatrix}
	1/4 &      0&    &        &      & \\
	0& 1     &-1/2&        &      &  \\
	&-1/2   &  1 &        &      &    \\
	&       &    &  \ddots&  -1/2&       \\
	&       &    &  -1/2  &  1   & -1/2  \\
	&       &    &        &-1/2  & 1/2     \\
	\end{pmatrix} \\
	\end{align}
\end{scriptsize}
	\noindent Note the first matrix (i.e. the $2\times 2$ block) is semi-positive definite (with eigenvalues $\lambda\in \{0, 5/4\}$), hence the following inequality holds:
\begin{scriptsize}
	\begin{align}
	\begin{pmatrix}
1/2 &       -1/2&    &        &      & \\
-1/2& \mathbf{2}&-1/2&        &      &  \\
&-1/2       &  1 &        &      &    \\
&           &    &  \ddots&  -1/2&       \\
&           &    &  -1/2  &  1   & -1/2  \\
&           &    &        &-1/2  & 1/2     \\
	\end{pmatrix} &\geq
	\begin{pmatrix}
	1/4 &      0&    &        &      & \\
0& 1     &-1/2&        &      &  \\
&-1/2   &  1 &        &      &    \\
&       &    &  \ddots&  -1/2&       \\
&       &    &  -1/2  &  1   & -1/2  \\
&       &    &        &-1/2  & 1/2     \\
	\end{pmatrix}
	\end{align}
	
\end{scriptsize}

	We then note that the bottom-right block of the matrix is $\Delta^{(T-1)} + \frac{1}{2}\ket{1}\bra{1}$. From Lemma \ref{Lemma:Tridiagonal_Partial_Penalty} this has a minimum eigenvalue $\lambda_0^{(k=2)} = 1-\cos(\frac{\pi}{2T-1}) > \lambda_0^{(k=1)} = 1-\cos(\frac{\pi}{2T})$. Thus the $k=2$ case has a larger minimum eigenvalue than the $k=1$ case.
	\\
	\\
	\noindent
	$\mathbf{k=3}:$ \newline
	\\
	The $k=3$ case follows similarly: 
	\begin{scriptsize}
	
	\begin{align}
	\begin{pmatrix}
	1/2& -1/2&&& \\
	-1/2& 1&-1/2&&  \\
	&-1/2&  \mathbf{2}&  -1/2 &      \\
	&&-1/2  &1   &      \\
	&&&&  \ddots&  -1/2    \\
	&&&&    -1/2  & 1/2      \\
	\end{pmatrix}
	\geq
	\begin{pmatrix}
	1/2 & -1/2&      &       & \\
	-1/2&  3/4&     0&       &  \\
	    &   0 &     1&  -1/2 &      \\
	    &     &-1/2  &   1   &      \\
	    &     &      &       &  \ddots  &  -1/2      \\
	    &     &      &       &    -1/2  & 1/2      \\
	\end{pmatrix}
	\end{align}
		
\end{scriptsize}
\noindent
	where we have used the same method as in equation (\ref{Eq:k=2_Case}). The top-left block is positive definite with eigenvalues $(5\pm \sqrt{17})/8>0.1$.
	From Lemma \ref{Lemma:Tridiagonal_Partial_Penalty} the minimum eigenvalue of the right-hand side is then $ = 1-\cos(\frac{\pi}{2T-3}) >\lambda_0^{(k=1)}$. 
	Hence
	\begin{align}
	\lambda_0^{(k=2)} \geq 1-\cos(\frac{\pi}{2T-3}) > \lambda_0^{(k=1)}.
	\end{align}
	\\
	\\
	\noindent
	$\mathbf{4 \leq k \leq \lfloor T/2 \rfloor }:$ \newline
	\\
	We now consider the case for $k \geq 4$. 
	For this we will consider the matrix $\Delta^{(T)} +  \ket{k}\bra{k}$ and split it into two block diagonal components: 
	\begin{scriptsize}

	\begin{align}
	&
	\begin{pmatrix}
	1/2& -1/2&&&& \\
	-1/2& 1&&&&  \\
	&&\ddots  & -1/2  &  &    \\
	&&-1/2&  1&  -1/2&   &  &  \\
	&&&    -1/2  & \mathbf{2}   & -1/2  \\
	&&&              &-1/2 & 1 & -1/2 \\
	&&&     &         &-1/2 & \ddots   \\
	\end{pmatrix} 
	= \begin{pmatrix}
	0& 0&&&& \\
	0& 0&&&&  \\
	&&\ddots  &  0 &  &    \\
	&&0&  1/4&   -1/2&   &  &  \\
	&&&      -1/2  & 1   & 0  \\
	&&&              & 0 & 0 & 0 \\
	&&&     &         &0 & \ddots   \\
	\end{pmatrix}
	\\
	+ &\begin{pmatrix}
	1/4& -1/4&&&& \\
	-1/4& 1/2&&&&  \\
	&&\ddots  & -1/4  &  &    \\
	&&-1/4&  1/2&   0&   &  &  \\
	&&&      0  & 1   & -1/2  \\
	&&&              &-1/2 & 1 & -1/2 \\
	&&&     &         &-1/2 & \ddots   \\
	\end{pmatrix} + 
	\begin{pmatrix}
	1/4& -1/4&&&& \\
	-1/4& 1/2&&&&  \\
	&&\ddots  & -1/4  &  &    \\
	&&-1/4&  1/4&   0&   &  &  \\
	&&&      0  & 0   & 0  \\
	&&&              & 0 & 0 & 0 \\
	&&&     &         &0 & \ddots   \\
	\end{pmatrix}.
	\end{align}
	\end{scriptsize}
	Alternatively, we can write this as a block matrix decomposition:
	
	\begin{footnotesize}
	\begin{align}
	\Delta^{(T)} + \ket{k}\bra{k} =& \begin{pmatrix}
	\frac{1}{2}\Delta^{(k-1)} + \frac{1}{4}\ket{k-1}\bra{k-1} & 0  \\
	0 & \Delta^{(T-k+1)} + \frac{1}{2}\ket{k}\bra{k}
	\end{pmatrix} 
	+ \begin{pmatrix}
	\frac{1}{2}\Delta^{(k-1)}  & 0  \\
	0 & 0
	\end{pmatrix}
	\\
	&+  
	\begin{pmatrix}
	0  & &&  \\
	& 1/4 & -1/2& \\
	& -1/2 & 1 & \\
	& &&& 0 
	\end{pmatrix}.
	\end{align}					
	\end{footnotesize}

	\noindent The second and third matrices are both positive semi-definite, and thus
\begin{small}
		\begin{equation} \label{Eq:Starting_Penalty:Block}
	\Delta^{(T)} + \ket{k}\bra{k} \geq \begin{pmatrix}
	\frac{1}{2}\Delta^{(k-1)} + \frac{1}{4}\ket{k-1}\bra{k-1} & 0  \\
	0 & \Delta^{(T-k+1)} + \frac{1}{2}\ket{k}\bra{k}
	\end{pmatrix}. 
	\end{equation}
\end{small}
	Without loss of generality, we can now restrict to $k \leq T/2$ (if $k> T/2$ we can do the same process as above, but swapping around the top-left and bottom-right blocks). We consider the two blocks separately:
	
	\paragraph{Top-Left Block:}
	~\newline
	
	From Lemma \ref{Lemma:Tridiagonal_Partial_Penalty} we find that the minimum eigenvalue of $\Delta^{(j)} + \frac{1}{2}\ket{j}\bra{j} $, denoted $\lambda_0^{(j)}$, is
	\begin{equation}
	\lambda_0^{(j)} = 1 - \cos(\frac{\pi}{2j+1}).
	\end{equation}
	\noindent 
	Hence we know that the smallest possible minimum eigenvalue of the top-left block of right-hand side of \ref{Eq:Starting_Penalty:Block} occurs when $j=\lfloor T/2 \rfloor$. Thus 
	\begin{align}
	 \lambda_0^{(\lfloor T/2\rfloor)} &= \frac{1}{2}\bigg(1 - \cos(\frac{\pi}{2\lfloor T/2 \rfloor+1})\bigg) \\
	 &> 1 - \cos(\frac{\pi}{2T}) = \lambda_0^{(k=1)}. 
	\end{align}
	
	\paragraph{Bottom-Right Block:}
	~\newline
	
	Again, we find the minimum eigenvalue of the bottom block is 
	\begin{align}
	\lambda_0 &= 1- \cos(\frac{\pi}{ 2(T-k+1)+1 }) \\
	 &> 1 - \cos(\frac{\pi}{2T})\\
	 &= \lambda_0^{(k=1)},
	\end{align}
	where the first to second line follows from the fact $4 \leq k\leq \lfloor T/2 \rfloor $.
	\newline	

	Thus $\min_{k} \lambda_0\left(\Delta^{(T)} +  \ket{k}\bra{k}\right)$ is achieved for $k=1, T$, for $T>T_0 = 3$. 
\end{proof}

\subsection{Endpoint Penalty Analysis}

In \cite{Caha_Landau_Nagaj} the spectra of quantum walks with end point penalties was considered. 
We do the same here and determine bounds for the ground state energies of these walks for a range of different strength penalties.
Our main object of study will be the Hamiltonian 
\begin{align}
H_T(\mu) = \Delta^{(T)}+ \mu \ket{T}\bra{T},
\end{align}
where we will explore how the minimum eigenvalue varies as a function of $T, \mu$.

\begin{lemma}
	The eigenvalues of $H_T(\mu)$ are the solutions of the equation
	\begin{align}
	\sqrt{\frac{\lambda}{\lambda-2}}\frac{y_T(\lambda)-x_T(\lambda)}{y_T(\lambda)+x_T(\lambda)} = \frac{\mu}{1-\mu},
	\end{align}
	where 
	\begin{align}
	x_T(\lambda)&= (1-\lambda - \sqrt{\lambda(\lambda-2)})^T      \\
	y_T(\lambda)&= (1-\lambda + \sqrt{\lambda(\lambda-2)})^T.
	\end{align}
\end{lemma}
\begin{proof}
	This follows from a standard recurrence relation for tridiagonal matrices: consider the characteristic equation $\det(H-\lambda\mathds{1})=0$. We can use a standard continuant recurrence relation:
	\begin{align}
	f_0 &= 1, \\
	f_1 &= 1/2-\lambda, \\
	f_2 &= (1-\lambda) f_1 - (1/4)f_0, \\
	&\vdots \\
	f_n &= (1/2+\mu -\lambda)f_{n-1} - (1/4)f_{n-1}.
	\end{align} 
	Solving this gives the characteristic equation
	\begin{align}
	p_T(\lambda)&= \frac{-2^{T-1}}{\sqrt{\lambda-2}}\big((\mu-1)x_T(\lambda)\sqrt{\lambda}+ \mu y_T(\lambda)\sqrt{\lambda-2}\big) \\
	&=0.
	\end{align}
	Rearranging gives the formula as in the lemma statement.
\end{proof}

Using the above lemma, we now prove properties of the eigenvalues of $H_T(\mu)$: 
\begin{theorem}
	 For $\mu = k/T$, the minimum eigenvalue of $H_T(\mu)$ is bounded by
	\begin{align}
	\lambda_0\bigg(H_T\bigg(\frac{k}{T}\bigg)\bigg) = \Theta\bigg(\frac{k}{T^2}\bigg),
	\end{align}
	for $k=O(1)$ and sufficiently large $T$. 

	Furthermore, for all $\mu$ and $m\in \mathds{N}$, $1\leq m\leq T$, 
	\begin{align}
	\lambda_m(H_T(\mu)) \geq \mu\bigg( 1- \cos(\frac{(2m-1)\pi}{2T})\bigg).
	\end{align}
\end{theorem}
\begin{proof}
	We first take the characteristic equation and consider $\mu=0,1/2,1$ values. 
	The eigenvalues corresponding to these values are known analytically by \cite{Yueh_Cheng_2008}\cite{Yueh_2005}. 
	Rearranging the characteristic equation gives
	\begin{align} \label{Eq:g_function}
	g_T(\lambda)=\sqrt{\frac{\lambda}{\lambda-2}}\frac{y_T(\lambda)-x_T(\lambda)}{y_T(\lambda)+x_T(\lambda)} = \frac{\mu}{1-\mu},
	\end{align}
	which can be equivalently written as 
	\begin{align} 
	\frac{p_T^{(\mu=0)}(\lambda)}{p_T^{(\mu=1)}(\lambda)} = \frac{\mu}{1-\mu}.
	\end{align}
	A sketch of $g_T(\lambda)$ can be seen in figure \ref{Fig:Characterstic_Ratio} for $T=7$.
	
	The eigenvalues for $\mu=0,1$ are known known analytically, hence it is known $p_T^{(\mu=1)}(\lambda)$ has zeros at $\lambda_k=1-\cos(\frac{(2k-1)\pi}{2T})$ for $k=1,2,\dots, T$, and $p_T^{(\mu=0)}(\lambda)$ has zeros at $\lambda_k= 1 - \cos(\frac{2k\pi}{(2T+1)})$, $k=0,1,\dots, T-1$.
	From this, we find that $g_T(\lambda)$ has poles at $\lambda_k=1-\cos(\frac{(2k-1)\pi}{2T})$ for $k=1,2,\dots, T$, and zeros at $\lambda_k= 1 - \cos(\frac{2k\pi}{(2T+1)})$, $k=0,1,\dots, T-1$ which are the eigenvalues of the $\mu=0$ case.
	
	Furthermore, we know the eigenvalues of the $\mu=1/2$ case and know that this occurs when the left-hand side of equation \ref{Eq:g_function} is equal to one.
	The eigenvalues for this case are $\gamma_k = 1- \cos(\frac{(2k-1)\pi}{2T+1})$, $k=1,2,\dots,T$, hence $g(\gamma_k)=1$ \cite{Yueh_2005}. 
	
	Hence we know $g_T(\lambda)$ must look like Fig. \ref{Fig:Characterstic_Ratio}. 
	By examining this, we can trivially put an upper bound on many of the eigenvalues, including the smallest, as $O(T^{-2})$.

	\begin{figure}[h!] 
		\centering
		\includegraphics[width=1\textwidth]{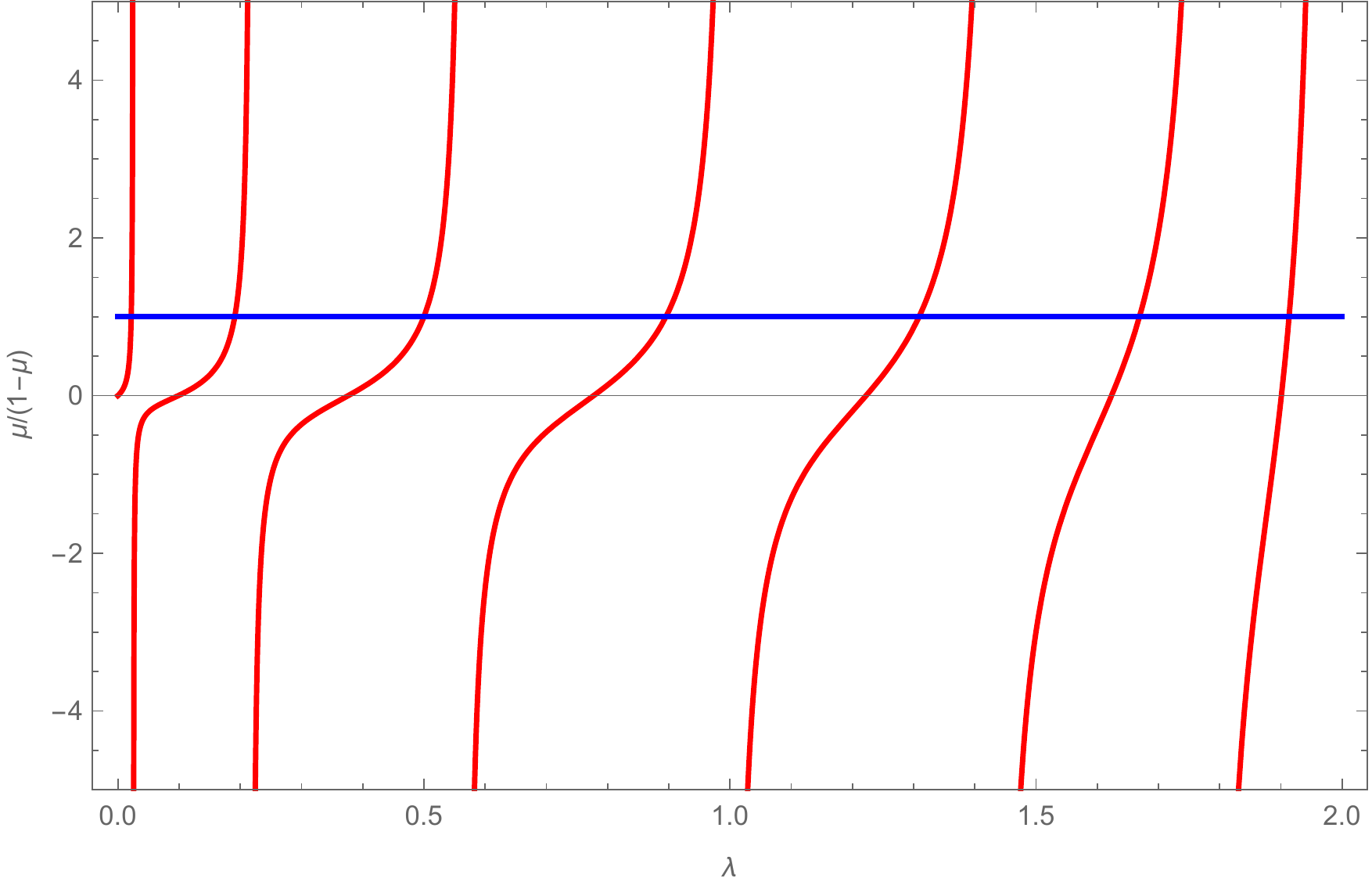}
		\caption{The red line represents function $g_7(\lambda)$. 
		The blue horizontal line corresponds to $\mu=1/2$.}
		\label{Fig:Characterstic_Ratio}
	\end{figure}

	To find solutions we consider where $g_T$ intersects horizontal lines of $\mu/(1-\mu)$ (see Fig. \ref{Fig:Characterstic_Ratio} for an example). 
	In particular we are interested in the case of large $T$ when $\mu$ is small. 
	We consider the cases where $\mu=k/T$ for some $k=O(1)$. 
	
	We consider the expansion of $g_T$ around $T= \infty$ in powers of $1/T$ to get
	\begin{align}
	g_T\bigg(\frac{k}{T^2}\bigg) = \sqrt{\frac{k}{2}}\tan(\sqrt{2k})\frac{1}{T} + O(T^{-2}).
	\end{align}
	We want to determine when this is equal to $\mu/(1-\mu) = k/(T-k)$. 
	Hence for sufficiently large $T$ and $k=O(1)$, we can write 
	\begin{align}
	c\frac{k}{T}< g_T\bigg(\frac{k}{T^2}\bigg)< c'\frac{k}{T}
	\end{align} 
	for some $c, c'=O(1), \  c'>c$. This gives 
	\begin{align}
	\lambda_0\bigg(H_T\bigg(\frac{k}{T}\bigg)\bigg) = \Theta\bigg(\frac{k}{T^2}\bigg).
	\end{align}
	

	\paragraph{Lower Bound}
	
	Finally we lower bound the eigenvalues:
	\begin{align}
	\Delta^{(T)} + \mu\ket{T}\bra{T} &= \mu(\Delta^{(T)}+ \ket{T}\bra{T}) - (1-\mu)\Delta^{(T)} \\
	&\geq \mu(\Delta^{(T)}+ \ket{T}\bra{T})
	\end{align}
	Hence, using the above inequality, we have that 
	\begin{align}
	\lambda_m(\Delta^{(T)} + \mu\ket{T}\bra{T}) &\geq \mu \lambda_m(\Delta^{(T)} + \ket{T}\bra{T}) \\
	&\geq \mu\bigg(1 - \cos(\frac{(2m-1)\pi}{2T})\bigg). 
	\end{align}
	
\end{proof}
Although we do not prove it, numerical analysis suggests that the gradient of $g_\mu(\lambda)$ at $\mu=1/2$ scales as $O(T^3)$. 
This suggests that for general $0<\mu<1/2$ and large $T$, the best bound achievable is $\lambda_0(H_T(\mu))=O(\mu^{T^{-3}}/T^2)$. 

We note the first bound in the lemma statement is an extension on \cite{Caha_Landau_Nagaj} where it was shown that $\lambda_0(H_T(\mu))=O(\mu/T)$. 
While this would give us the same upper bound as our result, it does not give the same lower bound. 
We note that the lower $\Omega(k/T^2)$ result here is only true for $\mu=k/T, \ k=O(1)$ and does \emph{not} apply for the case $\mu=O(1)$ in general.

\section{Hamiltonian Analysis for Standard Form Hamiltonians} \label{Sec:Standard_Form_Hamiltonian_Analysis}

In this section we consider Hamiltonians which encode computation with uniform weight transition rules. 
We also restrict ourselves to computations which do not branch --- given a basis state there is at most one transition rule that applies to it. 
We further restrict ourselves to the analysis of computations which do not branch into multiple tracks (such as those described by unitary labelled graphs as per \cite{Bausch_Cubitt_Ozols_2016}).

We first consider computations which have a deterministically accepted or rejected output: Exact-QMA (EQMA), as define in Def. \ref{Def:EQMA}. 
This will give us the relevant tools for examining Hamiltonians that encode QMA instances.

\subsection{Standard Form Hamiltonians}

The Hamiltonians we are interested in will fit into specific class of Hamiltonians which we call ``standard form Hamiltonians''. 
The idea will be that we encode the verification of problem instances in these Hamiltonians in a history state construction, as per \cite{Kitaev_Shen_Vyalyi}. Thus, as usual for these constructions, this class of Hamiltonians will contain three types of terms. 
The first is \emph{transition terms} which force the evolution from one state to the next. 
The second type of terms are \emph{penalty terms} which act to assign an energy penalty to any states which are not allowed by the computation. 
The final set of terms are \emph{computational penalty terms} which penalise computations that are not correctly initialised or result in a NO instance after the computation has been run.
We label this class `Standard-Form Hamiltonians', the definition of which is a modification to the class of the same name from \cite{Cubitt_Perez-Garcia_Wolf}.

\begin{definition}[Standard Basis States, from Section 4.1 of \cite{Cubitt_Perez-Garcia_Wolf}]
	Let the single site Hilbert space be $\mathcal{H}=\otimes_i \mathcal{H}_i$ and fix some orthonormal basis for the single site Hilbert space. Then a \emph{Standard Basis State} for $\mathcal{H}^{\otimes L}$ are product states over the single site basis.
\end{definition}
\noindent We now define standard-form Hamiltonians --- extending the definition from \cite{Cubitt_Perez-Garcia_Wolf}:
\begin{definition}[Standard-form Hamiltonian, definition extended from \cite{Cubitt_Perez-Garcia_Wolf}]
	\label{Def:Standard-form_H}
	We say that a Hamiltonian $H = H_{trans} + H_{pen} + H_{in} + H_{out}$ acting on a Hilbert space $\mathcal{H} = (\mathds{C}^C\otimes\mathds{C}^Q)^{\otimes L} = (\mathds{C}^C)^{\otimes L}\otimes(\mathds{C}^Q)^{\otimes L} =: \mathcal{H}_C\otimes\mathcal{H}_Q$ is of standard form if $H_{trans,pen, in, out} = \sum_{i=1}^{L-1} h_{trans,pen, in, out}^{(i,i+1)}$, and the local interactions $h_{trans,pen, in, out}$ satisfy the following conditions:
	\begin{enumerate}
		\item $h_{trans} \in \mathcal{B}\left((\mathds{C}^C\otimes\mathds{C}^Q)^{\otimes 2}\right)$ is a sum of transition rule terms, where all the transition rules act diagonally on $\mathds{C}^C\otimes\mathds{C}^C$ in the following sense. Given standard basis states $a,b,c,d\in\mathds{C}^C$, exactly one of the following holds:
		\begin{itemize}
			\item there is no transition from $ab$ to $cd$ at all; or
			\item $a,b,c,d\in\mathds{C}^C$ and there exists a unitary $U_{abcd}$ acting on $\mathds{C}^Q\otimes\mathds{C}^Q$ together with an orthonormal basis $\{\ket{\psi_{abcd}^i}\}_i$ for $\mathds{C}^Q\otimes\mathds{C}^Q$, both depending only on $a,b,c,d$, such that the transition rules from $ab$ to $cd$ appearing in $h_{trans}$ are exactly $\ket{ab}\ket{\psi^i_{abcd}}\rightarrow \ket{cd}U_{abcd}\ket{\psi^i_{abcd}}$ for all $i$. There is then a corresponding term in the Hamiltonian of the form
			$(\ket{cd}\otimes U_{abcd} - \ket{ab})(\bra{cd}\otimes U_{abcd}^\dagger - \bra{ab})$. 
		\end{itemize}
		\label{standard-form_H:transition_terms}
		\item $h_{pen} \in \mathcal{B}\left((\mathds{C}^C\otimes\mathds{C}^Q)^{\otimes 2}\right)$ is a sum of penalty terms which act non-trivially only on $(\mathds{C}^C)^{\otimes 2}$ and are diagonal in the standard basis, such that $h_{pen} = \sum_{(ab) \ Illegal  } \ket{ab}_C\bra{ab}\otimes \mathds{1}_{Q}$, where $(ab)$ are members of a disallowed/illegal subspace. 
		\label{standard-form_H:penalty_terms}
		\item $h_{in}=\sum_{ab} \ket{ab}\bra{ab}_C \otimes \Pi_{ab}$, where $\ket{ab}\bra{ab}_C \in (\mathds{C}^C)^{\otimes2}$ is a projector onto $(\mathds{C}^C)^{\otimes 2} $ basis states, and $\Pi_{ab}^{(in)} \in (\mathds{C}^{Q})^{\otimes 2}$ are orthogonal projectors onto $(\mathds{C}^{Q})^{\otimes 2}$ basis states. 
		\item $h_{out}= \ket{xy}\bra{xy}_C \otimes \Pi_{xy}$, where $\ket{xy}\bra{xy}_C \in (\mathds{C}^C)^{\otimes2}$ is a projector onto $(\mathds{C}^C)^{\otimes 2} $ basis states, and $\Pi_{xy}^{(in)} \in (\mathds{C}^{Q})^{\otimes 2}$ are orthogonal projectors onto $(\mathds{C}^{Q})^{\otimes 2}$ basis states. 
	\end{enumerate}
\end{definition}

We note that although $h_{out}$ and $h_{in}$ have essentially the same form, they will play a different role later in the proof.
We reserve $h_{out}$ for penalties applied to the output of the QTM only.

\subsection{Standard-Form Hamiltonian Analysis}

In this section we exactly determine the minimum eigenvalues of a standard-form Hamiltonian which has a ground state that encodes the verification of either a YES or NO EQMA problem instance. 
This is then used to bound the minimum eigenvalues for Hamiltonians encoding QMA computations. 
We begin with several definitions:
\begin{definition}[Legal and Illegal Pairs and States, from \cite{Cubitt_Perez-Garcia_Wolf}] \label{Def:Legal_and_Illegal}
	The pair $ab$ is an \emph{illegal pair} if the penalty term $\ket{ab}\bra{ab}_C\otimes \mathds{1}_Q$ is in the support of the $H_{pen}$ component of the Hamiltonian. If a pair is not illegal, it is legal. We call a standard basis state \emph{legal} if it
	does not contain any illegal pairs, and illegal otherwise.
\end{definition}

Then the following is a straightforward extension of Lemma 42 of \cite{Cubitt_Perez-Garcia_Wolf} with $H_{in}$ and $H_{out}$ terms included.
\begin{lemma}[Invariant subspaces, extended from Lemma 42 of \cite{Cubitt_Perez-Garcia_Wolf}]
	\label{Lemma:Invariant_subspaces}
	
	Let $H_{trans}$, $H_{pen}$, $H_{in}$ and $H_{out}$ define a standard-form Hamiltonian as defined in Def. \ref{Def:Standard-form_H}. Let $\mathcal{S}=\{S_i\}$ be a partition of the standard basis states of $\mathcal{H}_C$ into minimal subsets $S_i$ that are closed under the transition rules (where a transition rule $\ket{ab}_{CD}\ket{\psi} \rightarrow \ket{cd}_{CD}U_{abcd}\ket{\psi}$ acts on $\mathcal{H}_C$ by restriction to $(\mathds{C}^C)^{\otimes 2}$, i.e. it acts as $ab \rightarrow cd$).	
	Then $\mathcal{H} = (\bigoplus_S \mathcal{K}_{S_i})\otimes\mathcal{H}_Q$ decomposes into invariant subspaces $\mathcal{K}_{S_i}\otimes\mathcal{H}_Q$ of $H = H_{pen} + H_{trans} + H_{in} + H_{out}$ where $\mathcal{K}_{S_i}$ is spanned by $S_i$.
\end{lemma}

This is useful as it allows us to divide up the Hilbert space of $H_{trans}+H_{pen}+H_{in} + H_{out}$ into invariant subspaces in which each state has at most one transition applied to it in the forwards and backwards directions. We can then study the minimum eigenvalues of these subspaces separately. We use a modified version of the Clairvoyance Lemma from \cite{Cubitt_Perez-Garcia_Wolf} to do this.

However, before we can prove this, we need to introduce the following lemma that allows us to put a projector in a stoquastic form:
\begin{lemma}[Lemma 20 of \cite{Bausch_Crosson_2018}] \label{Lemma:Make_Projector_Stoquastic}
	Let $M \in \mathds{C}^{d\times d}$ be a projector. Then for any $1 \leq s < d$, there exists a block-diagonal
	unitary $V = V^\prime \oplus V^{\prime \prime}$ with $\dim V^\prime = s$ such that $D = V^\dagger M V$ is stoquastic. Furthermore, $V$ can be
	chosen such that, if we denote the rank of the upper-left $s \times s$ block with $r_a$ and the complementary
	lower-right block rank with $r_b$, then $D_{ij} \neq 0$ if and only if
	
	\begin{enumerate}
		\item $i=j$ and $j\leq r_a$
		
		\item \textit{or} $s\leq i = j \leq s + r_b$
		
		\item \textit{or} $ s \leq j=s +i \leq \min\{r_a, r_b\} $ 
		
		\item \textit{or} $s \leq i = s +j \leq \min\{r_a, r_b \}$.
	\end{enumerate}
\end{lemma}
\noindent
More intuitively, we can write
\begin{align}
V^\dagger M  V = D =\begin{pmatrix}
D_{aa} & D_{ab} \\
D_{ba} & D_{bb} 
\end{pmatrix},
\end{align}
where $D_{aa}$   and $D_{bb}$ are real, diagonal matrices with rank $r_a$ and $r_b$ respectively, and dimensions $s \times s$ and $(\abs{H_Q} -s) \times (\abs{H_Q} -s)$ respectively. Furthermore, $D_{ab}$ and $D_{ba}$ are real, negative, diagonal matrices with $\rank D_{ab} = \rank D_{ba}= \min\{r_a, r_b \}$. 
Its form can be seen explicitly as the red part in figure \ref{Fig:Conjugated_Hamiltonian}.

Our statement of the modified version of the Clairvoyance Lemma is 
\begin{lemma}[Clairvoyance Lemma, extended from Lemma 43 of \cite{Cubitt_Perez-Garcia_Wolf}]\label{Lemma:Clairvoyance}
	Let $H = H_{trans} + H_{pen} +H_{in} + H_{out}$ be a standard-form Hamiltonian, as defined in Def. \ref{Def:Standard-form_H}, and let $\mathcal{K}_S$ be defined as in Lemma \ref{Lemma:Invariant_subspaces}. Let $\lambda_0(\mathcal{K}_S)$ denote the minimum eigenvalue of the restriction $H\vert_{\mathcal{K}_S\otimes \mathcal{H}_Q}$ of $H = H_{trans} + H_{pen}+ H_{in} + H_{out}$ to the invariant subspace $\mathcal{K}_S \otimes\mathcal{H}_Q$.
	
	Assume that there exists a subset $\mathcal{W}$ of standard basis states for $\mathcal{H}_C$ with the following properties: 
	\begin{enumerate}
		\item All legal standard basis states for $\mathcal{H}_C$ are contained in $\mathcal{W}$. \label{clairvoyance:legal}
		\item $\mathcal{W}$ is closed with respect to the transition rules.  \label{clairvoyance:closure}
		\item  At most one transition rule applies in each direction to any state in $\mathcal{W}$. Furthermore, there exists an ordering on the states in each $S$ such that the forwards transition (if it exists) is from $\ket{t} \rightarrow \ket{t+1}$ and the backwards transition (if it exists) is $\ket{t}\rightarrow \ket{t-1}$.
		\label{clairvoyance:single_transition}
		\item For any subset $S\subseteq\mathcal{W}$ that contains only legal states, there exists at least one state to which no backwards transition applies and one state to which no forwards transition applies. 
		Furthermore, the unitaries associated with each transition $\ket{t} \rightarrow \ket{t+1}$ are $U_t=\mathds{1}_Q$, for $0\leq t\leq T_{init}-1$. Also, both the final state $\ket{T}$, and whether a state $\ket{t}$ has $t\leq T_{init}$, is detectable by a 2-local, translationally invariant projector acting only on nearest neighbour qudits. 
		\label{clairvoyance:initial_state}
	\end{enumerate}
	
	\noindent Then each subspace $\mathcal{K}_S$ falls into one of the following categories:
	\begin{enumerate} 
		\item $S$ contains only illegal states, and $H\vert_{\mathcal{K}_S\otimes \mathcal{H}_Q} \geq \mathds{1}$. \label{Clv_Lemma:_Illegal}
		
		\item $S$ contains both legal and illegal states, and 
		\begin{equation}
		W^\dagger H\vert_{\mathcal{K}_S\otimes \mathcal{H}_Q}W \geq \bigoplus_i \big(\Delta^{(\abs{S})} + \sum_{\ket{k}\in K_i}\ket{k}\bra{k} \big)
		\end{equation}
		where  $\sum_{\ket{k}\in K_i}\ket{k}\bra{k} := H_{pen}\vert_{\mathcal{K}_S\otimes \mathcal{H}_Q}$ and $K_i$ is some non-empty set of basis states and $W$ is some unitary.  \label{Clv_Lemma:Evolve_to_Illegal}
		
		\item 	 $S$ contains only legal states, then there exists a unitary $R=W(\mathds{1}_C\otimes(X\oplus Y)_Q)$ that puts $H\vert_{\mathcal{K}_S\otimes \mathcal{H}_Q}$ in the form \label{Clv_Lemma:Legal}
		\begin{align}
		R^\dagger H\vert_{\mathcal{K}_S\otimes \mathcal{H}_Q} R = \begin{pmatrix}
		H_{aa} & H_{ab} \\
		H_{ab}^\dagger & H_{bb} 
		\end{pmatrix},
		\end{align}
		where, defining $G:=\supp\bigg( \sum_{t=0}^{T_{init}-1} \Pi_t^{(in)} \bigg)$ and $s:= \dim G$,  
		\begin{itemize}
			\item $X:G \rightarrow G$.
			\item $Y:G^c \rightarrow G^c$.
			\item $H_{aa}$ is an $s\times s$ matrix.
			\item  $H_{aa}, H_{bb} \geq 0$ and are rank $r_a, r_b$ respectively.
			\item  $H_{aa}$ has the form
			\begin{align}
			H_{aa} = \bigoplus_i \big( \Delta^{(|S|)} + \alpha_i\ket{|S|-1}\bra{|S|-1} \big) + \sum_{t=0}^{T_{init-1}} \ket{t}\bra{t}\otimes X^\dagger \Pi_t|_G X.
			\end{align}
			\item  $H_{bb}$ is a tridiagonal, stoquastic matrix of the form
			\begin{equation}
			H_{bb} = \bigoplus_i(\Delta^{(\abs{S})} + \beta_i \ket{|S|-1}\bra{|S|-1} ).
			\end{equation}
			\item  $H_{ab} = H_{ba}$ is a real, negative diagonal matrix with rank $\min\{r_a, r_b\}$.
			\begin{equation}
			H_{ab} = H_{ba} = \bigoplus_i \gamma_i\ket{|S|-1}\bra{|S|-1}.
			\end{equation} 
		\end{itemize}
		where either we get pairings between the blocks such that 
		\begin{small}
		\begin{align}
		\begin{pmatrix} 
		\alpha_i & \gamma_i \\
		\gamma_i & \beta_i 
		\end{pmatrix} = 
		\begin{pmatrix} 
		1-\mu_i & -\sqrt{\mu_i(1-\mu_i)} \\
		-\sqrt{\mu_i(1-\mu_i)} & \mu_i 
		\end{pmatrix}\quad  or 	\quad	
		\begin{pmatrix} 
		1 & 0 \\
		0 & 1 
		\end{pmatrix},
		\end{align}
		\end{small}
		for $0\leq\mu_i\leq 1$, or we get unpaired values of $\alpha_i=0,1$ or $\beta_i=0,1$ for which we have no associated value of $\gamma_i$.
	\end{enumerate}
\end{lemma}

\begin{proof}
	The case of type \ref{Clv_Lemma:_Illegal} subspaces is straightforward as $\bra{x}_C \bra{\psi}_Q H_{pen} \ket{x}_C \ket{\psi}_Q \geq 1$ for any illegal standard basis state $\ket{x}_C$. Thus, $H\vert_{\mathcal{K}_S\otimes \mathcal{H}_Q} \geq \mathds{1}$.
	
	We consider subspaces of type \ref{Clv_Lemma:Evolve_to_Illegal} and \ref{Clv_Lemma:Legal}. To begin with, we initially follow the analysis from \cite{Cubitt_Perez-Garcia_Wolf}. 
	Consider the directed graph of states in $\mathcal{W}$
	formed by adding a directed edge between pairs of states connected by transition
	rules. By assumption, only one transition rule applies in each direction to any state
	in $\mathcal{W}$, so the graph consists of a union of disjoint paths (which could be loops in
	case \ref{Clv_Lemma:Evolve_to_Illegal}). Minimality of $S$ (Lemma \ref{Lemma:Invariant_subspaces}) implies that $S$ consists of a single such
	connected path.
	
	Let $t = 0, \dots  , |S | - 1$ denote the states in $S$ enumerated in the order induced by
	the directed graph. $H_{trans}$ then acts on the subspace $\mathcal{K}_S \otimes H_Q$ as
	\begin{align} 
	H_{trans}|_{\mathcal{K}_S \otimes \mathcal{H}_Q} &= \sum_{t=0}^{T-1} \frac{1}{2}\left(
	\ket{t}\bra{t}\otimes \mathds{1} + \ket{t+1}\bra{t+1}\otimes \mathds{1}
	- \ketbra{t+1}{t}\otimes U_t - \ketbra{t}{t+1}\otimes U_t^\dagger
	\right) 
	\end{align}
	where  $T = \abs{S}  - 1$ if the path in S is a loop, otherwise $T = \abs{S}- 2$.
	Whether we have a loop or not, we have
	\begin{align}
	H_{trans}|_{\mathcal{K}_S \otimes \mathcal{H}_Q} &\geq \sum_{t=0}^{\abs{S}-2} \frac{1}{2}\left(
	\ket{t}\bra{t}\otimes \mathds{1} + \ket{t+1}\bra{t+1}\otimes \mathds{1}
	- \ketbra{t+1}{t}\otimes U_t - \ketbra{t}{t+1}\otimes U_t^\dagger
	\right)
	\\
	&:= H_{path}
	\end{align}
	Equality arises when the path is not a loop.
	Furthermore, the ordering of the states in $S$ means we can write 
	\begin{align}
	H_{in} &= \sum_{t=0}^{T_{init}-1} \ket{t}\bra{t}\otimes \Pi_t \\
	H_{out} &= \ket{T}\bra{T}\otimes \Pi_{out}.
	\end{align}
	Now define
	\begin{equation}
	W := \sum_{t=0}^{\abs{S}-2}\ket{t}\bra{t}\otimes \prod_{i=0}^{t}U_i^\dagger + \ket{\abs{S}-1}\bra{\abs{S}-1} \otimes \mathds{1}_Q.
	\end{equation}
	Standard results from \cite{Kitaev_Shen_Vyalyi}, \cite{Cubitt_Perez-Garcia_Wolf} give 	$W^\dagger H_{trans}|_{\mathcal{K}_S \otimes \mathcal{H}_Q} W=\Delta^{(\abs{S})} \otimes \mathds{1}_Q$. 
	Furthermore $W^\dagger H_{in}|_{K_S \otimes H_Q} W =  H_{in}|_{K_S \otimes H_Q}$.
	To see why this is the case, note $U_t = \mathds{1}_Q$ for $0 \leq t\leq T_{init}-1$ and hence $H_{in}|_{K_S \otimes H_Q}$ is preserved under the conjugation.   
	Using these relations we find:
	\begin{align} \label{Eq:Conjugated_WHW}
	W^\dagger H|_{\mathcal{K}_S \otimes \mathcal{H}_Q} W \geq \Delta^{(\abs{S})} \otimes \mathds{1}_Q + H_{pen}|_{K_S \otimes H_Q} + \sum_{t=0}^{T_{init}-1}(\ket{t}\bra{t}\otimes \Pi_t) +  \ket{T}\bra{T} \otimes U^\dagger \Pi_{out} U
	\end{align} 
	where we have defined $U :=\prod_{j=0}^{T-1}U_j$ and have written out $H_{in}, H_{out}$ explicitly. Again, equality holds when the path is not a loop.	We see that $W^\dagger H_{pen}|_{K_S \otimes H_Q}W = H_{pen}|_{K_S \otimes H_Q}$  as, by Def. \ref{Def:Standard-form_H}, $H_{pen}$ only acts non-trivially on $\mathcal{H}_C$ while the unitaries $U_{t}$ act non-trivially only on $\mathcal{H}_Q$. Additionally, $S$ is defined to be minimal. 
	We now consider subspaces of type \ref{Clv_Lemma:Evolve_to_Illegal} and \ref{Clv_Lemma:Legal} separately.
	\newline\newline
	\noindent
	\textbf{Type \ref{Clv_Lemma:Evolve_to_Illegal} Subspaces}
	\newline
	By definition computations in type \ref{Clv_Lemma:Evolve_to_Illegal} subspaces must evolve to an illegal state at some point and hence $H_{pen}|_{\mathcal{K}_S \otimes \mathcal{H}_Q}$ must have some support on any subspace of this type. Noting that the last two terms in expression \ref{Eq:Conjugated_WHW} are positive semi-definite hence we can remove them to give the following, 
	\begin{align} 
	W^\dagger H\vert_{\mathcal{K}_S\otimes \mathcal{H}_Q} W &\geq \Delta^{(\abs{S})}\otimes \mathds{1}_Q + H_{pen}\vert_{\mathcal{K}_S\otimes \mathcal{H}_Q}. \\
	W^\dagger H\vert_{\mathcal{K}_S\otimes \mathcal{H}_Q} W &\geq \bigoplus_i (\Delta^{(\abs{S})} + \sum_{\ket{k}\in K_i} \ket{k}\bra{k})
	\end{align}
	for some non-empty set of basis states $K_i$.
	\newline\newline
	\noindent
	\textbf{Type \ref{Clv_Lemma:Legal} Subspaces}
	\newline 
	
	By definition, all the states in type \ref{Clv_Lemma:Legal} subspaces are legal, and hence $H_{pen}|_{\mathcal{K}_S \otimes \mathcal{H}_Q}=0$ in this subspace. Furthermore, the states in $S$ cannot form a loop by condition \ref{clairvoyance:initial_state} in the lemma statement. Thus the Hamiltonian takes the form 
	\begin{align}
	W^\dagger H|_{\mathcal{K}_S \otimes \mathcal{H}_Q} W = \Delta^{(\abs{S})} \otimes \mathds{1}_Q + \sum_{t=0}^{T_{init}-1}(\ket{t}\bra{t}\otimes \Pi_t) +\ket{T}\bra{T}\otimes U^\dagger \Pi_{out} U
	\end{align} 
	
	We then use Lemma  \ref{Lemma:Make_Projector_Stoquastic} to define a unitary $(\mathds{1}_C\otimes V)$, where $V = X \oplus Y$, which puts $V^\dagger U^{\dagger} \Pi_{t} UV$ into the form described in Lemma \ref{Lemma:Make_Projector_Stoquastic}. We choose $X$ to be an $s \times s$ unitary with the same support as $\sum_{t=0}^{T_{init}-1}\Pi_{t}$. 
	This gives
	\begin{align} \label{Eq:Conjugated_Hamiltonian}
	(\mathds{1}_C\otimes V)^\dagger W^\dagger H|_{\mathcal{K}_S \otimes \mathcal{H}_Q}W(\mathds{1}_C\otimes V)  =& \Delta^{(\abs{S})} \otimes \mathds{1}_Q + \sum_{t=0}^{T_{init}-1}(\ket{t}\bra{t})\otimes (X\oplus 0)^\dagger (\Pi_t) (X \oplus 0)\\ &+   \ket{T}\bra{T} \otimes V^\dagger U^\dagger \Pi_{out} U V
	\end{align}
	where
	\begin{equation} \label{Eq:M_Definition}
	(UV)^\dagger \Pi_{out} UV= M =  
	\begin{pmatrix}
	M_{aa} & M_{ab} \\
	M_{ab}^{\dagger} & M_{bb} 
	\end{pmatrix},
	\end{equation}
	and $M_{aa}$ is an $s \times s$ matrix with the same support as $ \sum_{t=0}^{T_{init}-1} \Pi_t$. 
	\begin{figure}[h!] 
		\centering
		\includegraphics[width=1\textwidth]{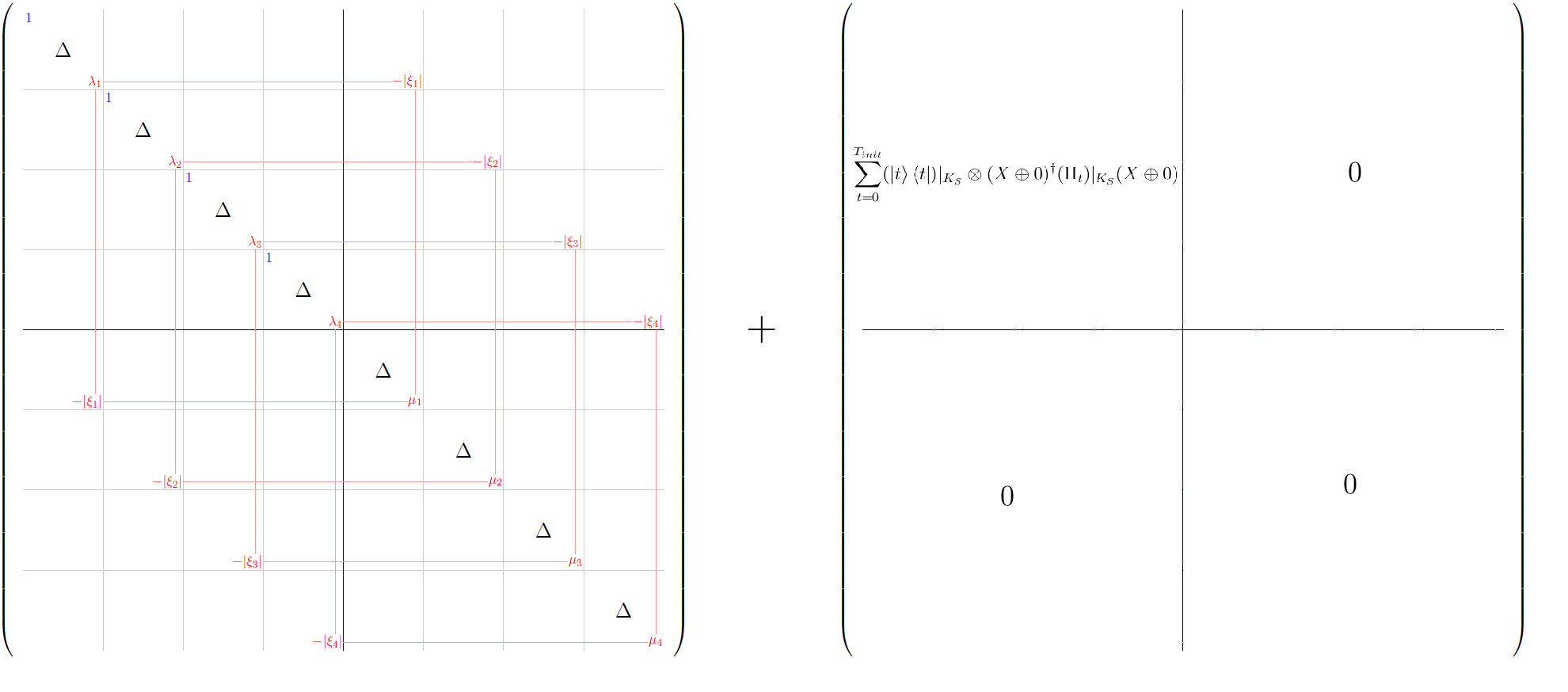}
		\caption{The conjugated Hamiltonian in restricted to a type \ref{Clv_Lemma:Legal} subspace. 
		The conjugated $H_{out}$ terms --- in red --- couple blocks in the support of $H_{in}$ to a block outside the support.	
		The conjugated $H_{in}$ terms --- represented by the term on the right --- couple $\Delta^{(T)}$ blocks in the $\Pi_{in}$ subspace together, where the coupling occurs on the first $T_{init}$ terms. 
		 }
		\label{Fig:Conjugated_Hamiltonian}
	\end{figure} 
	We can then decompose $(UV)^\dagger \Pi_{out} UV$ as a series of block diagonal terms 
	\begin{align} \label{Eq:Conjugated_Projector}
	(UV)^\dagger \Pi_{out} UV &= \bigoplus_i P_i, 
	\end{align}
	where either
	\begin{align}
	P_i &= \begin{pmatrix}
	\lambda_i & -|\xi_i| \\
	-|\xi_i| & \mu_i 
	\end{pmatrix}
	\end{align}
	or $P_i$ is equal to the $1\times 1$ matrices $P_i=1$ or $P_i=0$. Each non-zero $P_i$ must be a projector since $\Pi_{out}$ is a projector and is block diagonal in the $P_i$. The resulting conjugated Hamiltonian can be seen in figure \ref{Fig:Conjugated_Hamiltonian}. If $P_i$ is rank 2, then $P_i = \mathds{1}_2$ (as this is the only rank $2$, $2\times 2$ projector). If it is a rank 1, $2\times 2$ matrix then we can parametrise it in terms of a single value:
	\begin{align} \label{Eq:P_i_One_Parameter}
	P_i &= \begin{pmatrix}
	1-\mu_i & -\sqrt{\mu_i(1-\mu_i)} \\
	-\sqrt{\mu_i(1-\mu_i)} & \mu_i 
	\end{pmatrix}
	\end{align}
	(we see this from the fact $P_i$ is a rank 1 projector and hence we must be able to write it as $P_i = \ket{\chi}\bra{\chi}$ for  $\ket{\chi} = \sqrt{1-\mu_i}\ket{0}-\sqrt{\mu_i}\ket{1}$).
	This is the form as claimed in the lemma statement.
\end{proof}

\subsection{Encoding (E)QMA Verification in Standard Form Hamiltonians}

The Clairvoyance Lemma proves properties for a general standard form Hamiltonians under certain assumptions. 
We identify the basis states in $\mathcal{H}_C^{\otimes n}$ as ``clock states'', and as per Def. \ref{Def:Standard-form_H}, the transitions between these states are deterministic.
The clock states, together with the transition rules between them, form the \emph{clock}.
All penalties acting on the clocks states are diagonal.
There is a least one set of clock states for which there is an evolution which does not contain any illegal states and has a start and finish state: we called this a valid evolution.
Naturally we label the clock states in order as $\ket{0} \rightarrow \ket{1} \rightarrow \dots\rightarrow  \ket{t} \rightarrow \dots \rightarrow \ket{T}$.

To encode a Quantum Turing Machine in a Hamiltonian, we choose $U_t$ to correspond to a transition from the transition rule table, and a design a clock construction with an associated classical Hilbert space which labels the Quantum Turing Machine operations.
Although we will not explicitly state a clock construction, we note that constructions similar to the unary clock in \cite{Gottesman-Irani} and the base $\xi$ clock in \cite{Cubitt_Perez-Garcia_Wolf}both satisfy the assumptions in the Clairvoyance Lemma. 
In both these cases, if $H\in \mathcal{B}(\mathds{C}^d)^{\otimes L}$, then the computations these Hamiltonians encode have runtimes $O(L^2)$ and $O(L\xi^L\log(L))$ respectively, for some $\xi$ we are free to choose. 
These clocks have a dynamic initialisation and hence much of the previous analysis in \cite{Bausch_Crosson_2018} and \cite{Caha_Landau_Nagaj} does not apply to them.

We further introduce a set of other properties that these clocks share that we will find useful. We call this class of clocks ``Standard Form Clocks''.
\begin{assumption}[Standard Form Clock Properties]\label{Assumption:Clock_Properties}~\newline
	We assume the standard form clock construction for a standard form Hamiltonian $H\in \mathcal{B}(\mathds{C}^d)^{\otimes n}$ has the following properties:
	\begin{itemize}
		\item Satisfies assumptions 1-4 in the Clairvoyance Lemma (Lemma \ref{Lemma:Clairvoyance}).
		\item The total runtime of the clock is $T=T(n)=\Omega(n^2)$.
		\item For any set of basis states $S$ which contains at least one illegal state, then given any legal state $k \in S$, $k$ will evolve to an illegal state within $O(n)$ transitions. 
		\item The initialisation time is bounded as $T_{init}=O(n)=O(T^{1/2})$.
	\end{itemize}
	
\end{assumption}

\subsubsection{EQMA Computations}

We now consider a standard form Hamiltonian which will encode the time evolution of an EQMA verifier computation.
\begin{lemma}[EQMA Clairvoyance Lemma, extended from Lemma 43 of \cite{Cubitt_Perez-Garcia_Wolf}]\label{Lemma:EQMA_Clairvoyance}
	Let $H_{EQMA} = H_{trans} + H_{pen} +H_{in}+H_{out}$ be a Standard Form Hamiltonian encoding the verification of a EQMA problem instance  with a standard form clock. 
	Let the subspaces $\mathcal{K}_S$ be defined as in Lemma \ref{Lemma:Invariant_subspaces}. 
	Let $\lambda_0(\mathcal{K}_S)$ denote the minimum eigenvalue of the restriction $H\vert_{\mathcal{K}_S\otimes \mathcal{H}_Q}$ of $H_{EQMA} = H_{trans} + H_{pen}+ H_{in} + H_{out}$ to the invariant subspace $\mathcal{K}_S \otimes\mathcal{H}_Q$.
	
	\noindent Then each $\mathcal{K}_S$ falls into one of the following categories (corresponding to the same categories in Lemma \ref{Lemma:Clairvoyance}):
	\begin{enumerate} 
		\item $S$ contains only illegal states, and $\lambda_0(\mathcal{K}_S) \geq 1$. \label{EQMA_Clv_Lemma:_Illegal}
		
		\item $S$ contains both legal and illegal states, and $\lambda_0(\mathcal{K}_S) \geq 1 - \cos(\frac{\pi}{O(n)}) $. \label{EQMA_Clv_Lemma:Evolve_to_Illegal} 
		
		\item $S$ contains only legal states, and if the Hamiltonian's ground state encodes the verification of YES or NO instance, then \label{EQMA_Clv_Lemma:Legal}
		\[ \lambda_0(\mathcal{K}_S) = \begin{cases} 
		0 & YES \quad instance \\
		1-\cos(\frac{\pi}{2T}) & NO \quad instance,
		\end{cases}
		\]
		where $T(n)=\Omega(n^2)$ is the runtime of the standard form clock construction. 
		Furthermore, the Hamiltonian restricted to this subspace takes the form given in Lemma \ref{Lemma:Clairvoyance} for subspaces of type \ref{Clv_Lemma:Legal}, where $H_{ab}=H_{ba}=0$, and 
		\begin{equation}
		H_{bb} = \bigoplus_i(\Delta^{(T)} + \delta_i)
		\end{equation}
		where $\delta_i=1$ always for NO instances. For YES instances  $\delta_i=0$  for at least one case, and is $1$ otherwise. 
	\end{enumerate} 
\end{lemma}

\begin{proof} 

	By assumption, the standard form clock satisfies all the assumptions for the Clairvoyance Lemma (Lemma \ref{Lemma:Clairvoyance}) to hold, hence we can apply it straightforwardly.

	We now consider the three different types of subspaces as defined in the Clairvoyance Lemma.
	The case of subspace \ref{Clv_Lemma:_Illegal} follows straightforwardly from the result about subspace \ref{Clv_Lemma:_Illegal} in Lemma \ref{Lemma:Clairvoyance}.

	\textcolor{white}{a}
	\newline\newline \noindent
	\textbf{Type \ref{Clv_Lemma:Legal} Subspaces:} 
	\newline\newline
	Here we consider subspaces which contain only legal states (as per point \ref{Clv_Lemma:Legal} point 3 of the lemma statement).
	From the analysis in Lemma \ref{Lemma:Clairvoyance} that the Hamiltonian takes the form
	\begin{align} \label{Eq:Conjugated_Hamiltonian_2}
	R^\dagger H|_{\mathcal{K}_S \otimes \mathcal{H}_Q}R =& \Delta^{(\abs{S})} \otimes \mathds{1}_Q + R^\dagger \bigg(\sum_{t=0}^{T_{init}-1}(\ket{t}\bra{t})\otimes  (\Pi_t) \bigg) R\\ &+ \ket{T}\bra{T} \otimes \ D,
	\end{align}
	for $R=\mathds{1}_C\otimes (X\oplus Y)_Q$.
	
	We first see from Lemma \ref{Lemma:Clairvoyance} that the Hamiltonian can be broken up into four blocks, where the top left corresponds to $\supp\big(\sum_{t=0}^{T_{init}-1}\Pi_t\big)$.
	We now consider the YES and NO instances separately. 
	\newline \newline
	\noindent
	\textbf{EQMA NO Instances}
	\newline \newline
	We now consider the minimum eigenvalues in the case that the ground state encodes the verification of an NO  EQMA instance. We will need the following lemma from \cite{Bausch_Crosson_2018}:
	\begin{lemma}[Lemma 15 of \cite{Bausch_Crosson_2018}] \label{Lemma:U_NO_Instance}
		Let $U_t$ be the unitary representing the evolution of the computation at its $t^{th}$ step.
		Let $U=\prod_{t=0}^{T}U_t$ encode the verification of a NO instance.
		Then $D_{bb}$, defined as $R^\dagger \Pi_{out} R$ restricted to $\ker\sum_{t=1}^{T_{init}} \Pi_t$, has full rank. 
	\end{lemma}
	For a NO EQMA instance, the probability of any input being rejected is 1. 
	We now realise that for a NO instance the probability of the circuit rejecting a \textit{correctly initialised input} $\ket{x}\in \ker(\sum_{t=0}^{T_{init}=1}\Pi_t) \subset( \mathds{C}^{d})^{\otimes L}$ is $\bra{x}U^\dagger \Pi_{out}U\ket{x} = \bra{x^\prime}D_{bb}\ket{x^\prime} = 1 $ for an EQMA instance. 
	If we choose $\ket{x^\prime} \in  \ker H_{in}$ then this, in combination with $D_{bb}$ having maximum rank, implies $\mu_i = 1,\quad \forall i $.

	If we rearrange the rows and columns and write $D=\bigoplus_i P_i$, then $\mu_i=1$ implies either $P_i|_{G^c}=1$ or for $2\times 2$ matrices $P_i = \mathds{1}_2$ or
	\begin{equation}
	P_i = \begin{pmatrix}
	0 & 0 \\
	0 & 1
	\end{pmatrix}.
	\end{equation}
	The part of the conjugated Hamiltonian represented by $H_{bb}$ in the Clairvoyance Lemma (Lemma \ref{Lemma:Clairvoyance}) now decouples into a set of $T \times T$ blocks (see figure \ref{Fig:Conjugated_Hamiltonian}).
	We also note $H_{ab}=H_{ba}^\dagger =0$. 
	We now know that the lowest eigenvalue of $H_{bb}$ must belong to one of these separate $T \times T$ blocks, or lie in the upper-left block $H_{aa}$. 
	We now consider these two possibilities.

	These $T \times T$ blocks are penalised by $H_{out}$ \emph{only}; let $K_{out}$ represent one of these blocks, then $\supp(K_{out})\cap \supp(H_{in})=\emptyset$ (where $G$ is defined in Lemma \ref{Lemma:Clairvoyance}). 
	As a result they must take the form 
	\begin{equation}
	K_{out}= \Delta^{(T)} + \ket{T}\bra{T} = \begin{pmatrix}
	&  &  \\
	& \quad \boldsymbol{\Delta}^{(T)} &  \\
	&  & 1
	\end{pmatrix}.
	\end{equation}
	We now want to show that the subspace $\supp\bigg( \sum_{t=0}^{T_{init}-1} \Pi_t \bigg)$ (i.e. the subspace penalised by $H_{in}$) must have a minimum eigenvalue greater than or equal to the minimum eigenvalue of $K_{out}$. 
	To consider the blocks penalised by $H_{in}$ we first label $\supp \bigg( \sum_{t=0}^{T_{init}-1} \Pi_t \bigg) =: G$. We note
	 $D|_G= (1-\mu_1)\ket{1}\bra{1}+(1-\mu_2)\ket{2}\bra{2}+\dots  =  \bigoplus_i (1-\mu_i)$, where $\mu_i=0 \ or \ 1$. 
	 Then the Hamiltonian restricted to this subspace is
	
	\begin{align}
	R^\dagger H R|_G  &= \Delta^{(T)} \otimes \mathds{1}_G + (R^\dagger \sum_{t=0}^{T_{init}-1} \ket{t}\bra{t} \otimes  \Pi_{t} R)|_G + \ket{T}\bra{T}\otimes \bigoplus_i (1-\mu_i) \\
	&\geq \Delta^{(T)} \otimes \mathds{1}_G + (R^\dagger \sum_{t=0}^{T_{init}-1} \ket{t}\bra{t} \otimes  \Pi_{t} R)|_G 
	\end{align}
	We then recall that $R=W(\mathds{1}_C\otimes(X\oplus Y)_Q)$, hence we can conjugate with the inverse $\mathds{1}_C\otimes X_Q = R|_G$ to get
	\begin{align} 
	W^\dagger H W|_G& \geq \Delta^{(T)} \otimes \mathds{1}_G + \sum_{t=0}^{T_{init}-1} \ket{t}\bra{t} \otimes  \Pi_{t}|_G \label{Eq:H_in_Penalised_Block} \\
	&\geq  \bigoplus_{i=1}^s \bigg(  \Delta^{(T)} +  \sum_{\ket{k}\in Z_i}\ket{k}\bra{k} \bigg). \label{Eq:Deconjugate}
	\end{align} 
	where $Z_i$ is a non-empty set of basis elements corresponding to the elements penalised by $H_{in}$, for $0\leq k\leq T_{init}-1$.
	Going from equation (\ref{Eq:H_in_Penalised_Block}) to (\ref{Eq:Deconjugate}) we have used the fact that the Hamiltonian decomposes into blocks: one block in $G$ and the other in $G^c$. 
	Then the matrix $\bigoplus_i(1-\mu_i)$ is positive semi-definite.
	
	Note that the matrix in equation (\ref{Eq:Deconjugate}) is a block diagonal matrix with blocks of the form
	\begin{equation}
	K_{in}(Z_i) = \Delta^{(T)} + \sum_{\ket{k} \in Z_i} \ket{k}\bra{k}
	\end{equation}
	We now want to show that minimum eigenvalue of blocks of the form $K_{in}(Z_i)$ is larger than those of $K_{out}$ $\forall k$. To do this we use Lemma \ref{Lemma:Minimum_Eigenvalue} derived earlier

	First realise that $K_{in}(Z_i) \geq \Delta^{(T)}+\ket{j}\bra{j}$, where $j$ is the smallest integer such that $\ket{j}\in Z_i$.
	To see that a $K_{out}$ block always exists we note that it must be possible to choose a state $\ket{x}\in \ker H_{in}$ for a non-trivial kernel, and for a NO EQMA instance this must correspond to a $K_{out}$ block.
	From Lemma \ref{Lemma:Minimum_Eigenvalue} we see that the minimum eigenvalue of $\Delta^{(T)}+\ket{k}\bra{k}$ occurs for $k=1$, which is equal to the minimum eigenvalue of $K_{out}$ blocks.
	Hence $K_{in}(Z_i) \geq \Delta^{(T)}+\ket{j}\bra{j}\geq K_{out}$.
	
	From Lemma \ref{Lemma:Tridiagonal_Full_Penalty} we find that the eigenvalues of $K_{out}$ blocks are $1-\cos(\frac{(2m-1)\pi}{2T})$, for $m=1,2...,T$, thus giving a minimum eigenvalue of a NO EQMA instance as:
	\begin{equation}
	\lambda_0(\mathcal{K}_S) = 1 - \cos\bigg(\frac{\pi}{2T}\bigg)
	\end{equation}
	\newline\newline
	\noindent
	\textbf{EQMA YES Instances}
	\newline \newline
	For an EQMA YES instance we assume that $\ker\bigg( \sum_{t=0}^{T_{init}-1} \Pi_t\bigg)$ is non-trivial. Then, by definition, for a YES EQMA instance, there exists a state that is in  $\ker \bigg( \sum_{t=0}^{T_{init}}\Pi_{t} + U^\dagger \Pi_{out}U \bigg)$. Thus there exists a block with $\mu_i=0$, hence only contains $\Delta^{(T)}$. By standard analysis we see that the corresponding minimum eigenvalue eigenspace for YES instances is
	\begin{equation}
	\ker\left(H_{trans} + H_{pen} + H_{in} + H_{out}\right) =
	span \left\{
	\frac{1}{\sqrt{\abs{S}}}\sum_{t=0}^{\abs{S}-1}\ket{t}_C\ket{\psi_t}_Q
	\right\}
	\end{equation}
	where $\ket{t}_C$ are the states in $S$, $\ket{\psi_0}$ is any state in $\mathcal{H}_Q$, and $\ket{\psi_t} := U_t\dots U_1\ket{\psi_0}_Q$ where $U_t$ is the unitary on $\mathcal{H}_Q$ appearing in the transition rule that takes $\ket{t-1}_C$ to $\ket{t}_C$. These states have eigenvalue $0$. All other states in $\mathcal{K}_S\otimes\mathcal{H}_Q$ have energy at least that of the a NO instance.
	\newline\newline
	\noindent
	\textbf{Type \ref{Clv_Lemma:Evolve_to_Illegal} Subspaces:}
	\newline\newline
	We now consider subspaces which contain both legal and illegal states.
	From the Clairvoyance Lemma (Lemma \ref{Lemma:Clairvoyance}) we have that the Hamiltonian after conjugation takes the form:
	\begin{align} \label{Eq:Type_2_Penalties}
	W^\dagger H\vert_{\mathcal{K}_S\otimes \mathcal{H}_Q} W &\geq \bigoplus_i \bigg(\Delta^{(\abs{S})} + \sum_{\ket{k}\in Z_i} \ket{k}\bra{k} \bigg)
	\end{align}
	where $\sum_{\ket{k}\in Z_i} \ket{k}\bra{k} := H_{pen}\vert_{\mathcal{K}_S\otimes \mathcal{H}_Q}$. 
	We note that here each basis state within a $\mathcal{K}_S\otimes \mathcal{H}_Q$ block represents a different time step as the computation propagates.
	We want to lower-bound the energy of these $\mathcal{K}_S\otimes \mathcal{H}_Q$ subspaces such that they have energy larger than subspaces of type \ref{EQMA_Clv_Lemma:Legal}. 
	Thus we consider the clock properties assumption (Assumption \ref{Assumption:Clock_Properties}) which tells us that any state in subspace \ref{EQMA_Clv_Lemma:Evolve_to_Illegal} must reach an illegal state in $O(n)$ steps forwards or backwards. 
	Thus for $\ket{k}\in Z_i$ there must be another state $\ket{r}\in Z_i$, such that $r \leq k +O(n)$ (unless $\abs{S} < k+O(n)$) and similarly in the backwards direction.
	
	Now consider the right-hand side of \cref{Eq:Type_2_Penalties}. This can be decomposed into sets of shorter 1D walks of length $\ell_i \leq O(n)$ such that each of these shorter paths contains an penalty term in its first element, thus allowing us to write 
	\begin{align}
	\Delta^{(\abs{S})} + \sum_{\ket{k}\in Z_i} \ket{k}\bra{k} &= \bigoplus_j \bigg( \Delta^{(\ell_j)} + \ket{m_j}\bra{m_j} \bigg)+\sum E_j \\
	&\geq \bigoplus_j \bigg( \Delta^{(\ell_j)} + \ket{m_j}\bra{m_j} \bigg)
	\end{align}
	where $\ket{m_j}\in Z_i$ is some basis state corresponding to the first element of $\Delta^{(\ell_i)}$, and 
	\begin{equation}
	E_j = \begin{pmatrix}
	1/2 & -1/2 \\
	-1/2 & 1/2
	\end{pmatrix}.
	\end{equation}
	The inequality comes from the fact the terms $E_j$ are are positive semi-definite.
	Since $\ell_i \leq O(n)$, and using Lemma \ref{Lemma:Minimum_Eigenvalue}, we can bound the minimum eigenvalue of each of these matrices as
	\begin{align}
	\lambda_0(\mathcal{K}_S)\geq 1 - \cos\bigg( \frac{\pi}{O(n)} \bigg).
	\end{align}
	For sufficiently large $n$, this is always larger than the minimum eigenvalues of type-\ref{EQMA_Clv_Lemma:Legal} subspaces.
\end{proof}

So far we have found the minimum eigenvalue of the standard form Hamiltonian encoding the verification of an EQMA instances.
We now consider the form of the ground states themselves.
\begin{lemma}[EQMA Ground States] \label{Lemma:EQMA_Ground_States}
	Let $H\in\mathcal{B}(\mathds{C}^{\otimes n})$ be a standard form Hamiltonian as described in Def. \ref{Def:Standard-form_H}.
	Let the Hamiltonian encode the verification of an EQMA instance and define $\ket{\psi_t}=\prod_{j=0}^{t}U_t\ket{\psi_0}$, for some initial state $\ket{\psi_0}$. 
	Then the ground states for the YES and NO instances take the form
	\begin{align}
	\ket{\nu_{YES}} &= \frac{1}{\sqrt{T}} \sum_{t=0}^{T-1} \ket{t}\ket{\psi_t} \\
	\ket{\nu_{NO}} &=\sum_{t=0}^{T-1} \bigg(  \sin\bigg( \frac{(t+1)\pi}{2T}\bigg) - \sin\bigg( \frac{\pi t}{2T} \bigg) \bigg) \ket{t}\ket{\psi_t}, \\
	&= \sum_{t=0}^{T-1} 2\cos\bigg( \frac{(2t+1)\pi}{2T}\bigg)\sin\bigg(\frac{\pi t}{2T}\bigg)\ket{t}\ket{\psi_t}. 
	\end{align}
\end{lemma}
\begin{proof}
	We see that the ground state energies for these two cases correspond to Hamiltonians of the form of $T\times T$ matrices
	\begin{align}
	H^{(YES)}_{EQMA}\vert_{\mathcal{K}_S\otimes \mathcal{H}_Q} = \Delta^{(T)},
	\end{align}
	which is well known to have a uniform superposition as its ground state. 
	For the NO instance, we see
	\begin{align}
	H^{(NO)}_{EQMA}\vert_{\mathcal{K}_S\otimes \mathcal{H}_Q} = \Delta^{(T)} + \ket{T}\bra{T}.
	\end{align}
	The ground state of this is given in \cite{Yueh_Cheng_2008}.
\end{proof}

\subsection{QMA Computation}

We now consider a Standard-Form Hamiltonian encoding the verification of a QMA problem instance and its associated eigenvalues.
Before we do so we introduce the following lemma:
\begin{lemma}[Lemma 4 of \cite{Kempe_Kitaev_Regev}] \label{Lemma:KKR_Eigenvalue_Inequality}
	Let $H_1$ and $H_2$ be two Hamiltonians with eigenvalues $\mu_1 \leq \mu_2 \leq ...$ and $\sigma_1 \leq \sigma_2 \leq ...$. Then, for all $j$, $|\mu_j - \sigma_j | \leq || H_1 - H_2||$.
\end{lemma}

\begin{lemma}[QMA Clairvoyance Lemma]\label{Lemma:QMA_Clairvoyance}
	Let $H = H_{trans} + H_{pen} +H_{in}+H_{out}$ be a Standard-Form Hamiltonian encoding the evolution of a QMA verifier. 
	Let the subspaces $\mathcal{K}_S$ be defined as in Lemma \ref{Lemma:Invariant_subspaces} and let $\lambda_0(\mathcal{K}_S)$ denote the minimum eigenvalue of the restriction $H\vert_{\mathcal{K}_S\otimes \mathcal{H}_Q}$ of $H = H_{trans} + H_{pen}+ H_{in} + H_{out}$ to the invariant subspace $\mathcal{K}_S \otimes\mathcal{H}_Q$.
	
	\noindent Then each $\mathcal{K}_S$ falls into one of the following categories (corresponding to the same categories in Lemma \ref{Lemma:Clairvoyance}):
	\begin{enumerate} 
		\item $S$ contains only illegal states, and $\lambda_0(\mathcal{K}_S) \geq 1$. \label{QMA_Clv_Lemma:_Illegal}
		
		\item $S$ contains both legal and illegal states., and $\lambda_0(\mathcal{K}_S) \geq 1 - \cos(\frac{\pi}{8L}) $. \label{QMA_Clv_Lemma:Evolve_to_Illegal} 
		
		\item $S$ contains only legal states. \label{QMA_Clv_Lemma:Legal} 
		
		Define $\eta$ to be the probability of error for QMA, as per Def. \ref{Def:QMA}. 
		Let $H^{(YES/NO)}_{QMA}$ represents a Hamiltonian encoding the verification of a YES/NO instance, then its minimum eigenvalue is bounded by 
		\begin{align}
		0 \leq &\lambda_0\big( H_{QMA}^{(YES)}\vert_{\mathcal{K}_S\otimes \mathcal{H}_Q}\big) \leq \eta^{1/2} \\
		1-\cos\bigg(\frac{\pi}{2T}\bigg) - \eta^{1/2} \leq  &\lambda_0\big( H_{QMA}^{(NO)}\vert_{\mathcal{K}_S\otimes \mathcal{H}_Q}\big) \leq 1-\cos\bigg(\frac{\pi}{2T}\bigg)
		\end{align}
	\end{enumerate}
\end{lemma}
\begin{proof}
	The Hamiltonian is of standard form and hence satisfies the assumptions of the Clairvoyance Lemma (Lemma \ref{Lemma:Clairvoyance}).
	The statements for subspaces of types \ref{QMA_Clv_Lemma:_Illegal} and \ref{QMA_Clv_Lemma:Evolve_to_Illegal} follow directly from Lemma \ref{Lemma:Clairvoyance} by the same reasoning as the EQMA case in Lemma \ref{Lemma:EQMA_Clairvoyance}.
	
	We now consider subspaces of type \ref{QMA_Clv_Lemma:Legal}. To do this we define an ``EQMA Hamiltonian'' $\tilde{H}_{EQMA}$ that has the same eigenvalues as $H_{QMA}$ would have if all its verification computations were computed deterministically. We then bound the eigenvalues of type \ref{QMA_Clv_Lemma:Legal} subspaces relative to this EQMA Hamiltonian.
	
	\textcolor{white}{a} \\
	\noindent
	\textbf{Defining the EQMA Hamiltonian} \\
	\noindent
	First we consider a standard form Hamiltonian which encodes the verification a QMA instance, $H_{QMA}$. From it, we define an EQMA Hamiltonian which corresponds to it.
	Following from the Clairvoyance Lemma (Lemma \ref{Lemma:Clairvoyance}), we can conjugate $H_{QMA}\vert_{\mathcal{K}_S\otimes \mathcal{H}_Q}$ by the unitary $R=W(\mathds{1}_C\otimes (X\oplus Y)_Q)$ to put it in the following form:
	\begin{align}
	R^\dagger H_{QMA}\vert_{\mathcal{K}_S\otimes \mathcal{H}_Q}R =&  \Delta^{(T)} \otimes \mathds{1}_Q +  R^\dagger H_{in}\vert_{\mathcal{K}_S\otimes \mathcal{H}_Q}R\\ &+ \ket{T}\bra{T}\otimes  D^{QMA},
	\end{align}
	where $D^{QMA}=\bigoplus P_i^{QMA}$, where $P_i^{QMA}$ are the matrices described in the Clairvoyance (Lemma \ref{Lemma:Clairvoyance}).
	
	We then define the corresponding EQMA Hamiltonian to be:
	\begin{equation}
	R^\dagger \tilde{H}_{EQMA}\vert_{\mathcal{K}_S\otimes \mathcal{H}_Q}R := R^\dagger H_{QMA}\vert_{\mathcal{K}_S\otimes \mathcal{H}_Q}R - \ket{T}\bra{T}\otimes D^{QMA} + \ket{T}\bra{T}\otimes D^{EQMA}.
	\end{equation}
	where we will define $D^{EQMA}$ below.
	By rearranging the rows and columns it is possible to write $D^{(E)QMA}= \bigoplus_i P_i^{(E)QMA}$, and thus
	\begin{align}
	R^\dagger \tilde{H}_{EQMA}\vert_{\mathcal{K}_S\otimes \mathcal{H}_Q}R -  R^\dagger H_{QMA}\vert_{\mathcal{K}_S\otimes \mathcal{H}_Q}R =  \ket{T}\bra{T}\otimes\bigoplus_i (P_i^{EQMA} - P_i^{QMA})\vert_{\mathcal{K}_S\otimes \mathcal{H}_Q},
	\end{align}
	where $P_i^{EQMA}$ is of the same dimension as the corresponding $P_i^{QMA}$. 
	If $P_i^{QMA}$ is a $2\times 2$ matrix, then from the Clairvoyance Lemma (Lemma \ref{Lemma:Clairvoyance}) it is known that
	\begin{align} \label{Eq:P_i_One_Parameter_2}
	P_i^{QMA} &= \begin{pmatrix}
	1-\mu_i & -\sqrt{\mu_i(1-\mu_i)} \\
	-\sqrt{\mu_i(1-\mu_i)} & \mu_i 
	\end{pmatrix} \ or \ \mathds{1}.
	\end{align}
	If $P_i^{QMA}$ has dimension $1$, then $P_i^{QMA}=0, 1$. 
	If $P_i^{QMA}$ is rank 0, 1 or 2, then the corresponding $P_i^{EQMA}$ is chosen to also be rank 0, 1 or 2 and of the same dimensions.  From the definition of EQMA, we see that 
	\begin{equation}
	P_i^{EQMA} = \mathds{1} \quad or \quad \begin{pmatrix}
	0 & 0 \\
	0 & 1
	\end{pmatrix}
	\end{equation}
	for rank 2 and rank 1 $2\times 2$ matrices respectively.
	\newline
	\newline
	\noindent
	\textit{Aside: the Hamiltonian $\tilde{H}_{EQMA}$ defined here will generally be highly non-local. To understand this, we note $H_{EQMA}$ is defined with respect to a QMA verification circuit. However, this will not concern us since we only require that $R^\dagger \tilde{H}_{EQMA} R$ have a form and minimum eigenvalue close to $R^\dagger H_{QMA} R$ to allow us to analyse the spectrum of $H_{QMA}$ more easily.}
	
	We now analyse to subspaces of type \ref{QMA_Clv_Lemma:Legal} (as defined in Lemma \ref{Lemma:Clairvoyance}).
	\newline
	\newline
	\noindent
	\textbf{Proving Energy Separation for YES and NO instances.} \\
	\noindent
	We now consider how the $\{\mu_i\}_i$ relate to the probability of witness rejection. 
	\newline
	\newline
	\noindent\textbf{NO Instances:}
	\newline
	The probability of a correctly initialised input $\ket{x}\in \ker\sum_{t=0}^{T_{init}}\Pi_t = G^c \subset ( \mathds{C}^{d})^{\otimes n}$ being rejected is $\bra{x}\Pi_{out}|_{G^c}\ket{x} =  \bra{x^\prime} \bigoplus_i P_i^{QMA}|_{G^c}\ket{x^\prime} \leq 1 $. From the definition of QMA we have $1-\eta \leq \bra{x^\prime} \bigoplus_i P_i^{QMA}\ket{x^\prime} \leq 1$ for a NO QMA instance. 
	We note $P_i^{QMA}|_{G^c} = \bigoplus_i \mu_i$. 
	This implies $1-\eta \leq \mu_i \leq 1$ for a Hamiltonian encoding the verification of a NO instance.
	
	We now see that if $P^{QMA}_i$ is one dimensional, $P^{QMA}_i= P^{EQMA}_i =1$ for NO instances. If $P^{QMA}_i$ is a $2\times 2$ matrix, then the $P_i^{QMA}$ is as in expression \ref{Eq:P_i_One_Parameter_2}. Thus $P_i^{EQMA} - P_i^{QMA} = 0$ or 
	\begin{equation}
	(P_i^{EQMA} - P_i^{QMA})\vert_{\mathcal{K}_S\otimes \mathcal{H}_Q} =   
	\begin{pmatrix}
	\mu_i-1 & \sqrt{\mu_i(1-\mu_i)} \\
	\sqrt{\mu_i(1-\mu_i)} & 1-\mu_i 
	\end{pmatrix}.
	\end{equation}
	\newline 
	\newline
	\noindent\textbf{YES Instances:}
	\newline
	We now consider the similar argument for YES instances: 
	By the definition of QMA there must be at least one eigenvector $\ket{x}\in \ker(\sum_{t}\Pi_t)$ for which $0\leq \bra{x}\bigoplus_i U^\dagger P_i^{QMA}|_{G^c}U\ket{x} \leq \eta$. By the same reasoning we find there must be at least one $P_i^{QMA}=0$ or one $0 \leq \mu_i \leq \eta$. 
	
	The $P_i^{QMA}$ are different from the NO instances as at least one witness must be accepted. Hence we know in the EQMA case that the ground state can receive no energy penalty, hence 
	\begin{equation} \label{Eq:EQMA_Projector_Types}
	P_i^{EQMA} =  \begin{pmatrix}
	1 & 0 \\
	0 & 0
	\end{pmatrix} or
	\begin{pmatrix}
	0 & 0 \\
	0 & 0
	\end{pmatrix} or
	\begin{pmatrix}
	1 & 0 \\
	0 & 1
	\end{pmatrix}
	or
	\begin{pmatrix}
	0 & 0 \\
	0 & 1
	\end{pmatrix}.
	\end{equation}
	The first two of these appear for witnesses that are accepted by the verifier, and thus must occur for at least one $i$ (as we are considering a YES instance). The third and fourth appear for rejected witnesses. We now consider the corresponding QMA cases (for convenience we will label $\mu_i=\gamma_i$ for witness that are accepted, hence $0 \leq \gamma_i \leq \eta$):
	\begin{align} \label{Eq:QMA_Projector_Types}
	P_i^{QMA} =&  \begin{pmatrix}
	1-\gamma_i & -\sqrt{\gamma_i(1-\gamma_i)} \\
	-\sqrt{\gamma_i(1-\gamma_i)} & \gamma_i
	\end{pmatrix} or
	\begin{pmatrix}
	0 & 0 \\
	0 & 0
	\end{pmatrix} \\
	&or
	\begin{pmatrix}
	1 & 0 \\
	0 & 1
	\end{pmatrix}
	or
	\begin{pmatrix}
	1-\mu_i & -\sqrt{\mu_i(1-\mu_i)} \\
	-\sqrt{\mu_i(1-\mu_i)} & \mu_i
	\end{pmatrix}.
	\end{align}
	Here $\mu_i$ is only bounded between $1-\eta\leq\mu_i\leq 1$.
	If we now consider accepting witnesses we find either $P_i^{EQMA} - P_i^{QMA}=0$ or 
	\begin{align}
	P_i^{EQMA} - P_i^{QMA}= \begin{pmatrix}
	\gamma_i & \sqrt{\gamma_i(1-\gamma_i)} \\
	\sqrt{\gamma_i(1-\gamma_i)} & -\gamma_i
	\end{pmatrix}.
	\end{align}
	
	~\newline
	\noindent\textbf{Energy Bound:}
	\newline
	We now consider the difference for the NO case. 
	\begin{align}
	\tilde{H}_{EQMA}\vert_{\mathcal{K}_S\otimes \mathcal{H}_Q} - H_{QMA}\vert_{\mathcal{K}_S\otimes \mathcal{H}_Q} &\cong \ket{T}\bra{T} \otimes \bigoplus_i (P_i^{EQMA} - P_i^{QMA})\vert_{\mathcal{K}_S\otimes \mathcal{H}_Q} \\
	||\tilde{H}_{EQMA}^{(NO)}\vert_{\mathcal{K}_S\otimes \mathcal{H}_Q} - H_{QMA}^{(NO)}\vert_{\mathcal{K}_S\otimes \mathcal{H}_Q}|| &= \max_i||P_i^{EQMA} - P_i^{QMA}||  \\
	&= \max_i (1-\mu_i)^{1/2} \\
	&\leq \eta^{1/2}.
	\end{align}
	Using Lemma \ref{Lemma:KKR_Eigenvalue_Inequality} we get:
	\begin{align}
	|\lambda_0 \big(\tilde{H}_{EQMA}^{(NO)}\vert_{\mathcal{K}_S\otimes \mathcal{H}_Q}\big) -\lambda_0\big( H_{QMA}^{(NO)}\vert_{\mathcal{K}_S\otimes \mathcal{H}_Q}\big)| &\leq \eta^{1/2}.
	\end{align}
	
	We now consider YES instances. 
	To bound the minimum eigenvalue of Hamiltonians encoding the verification of YES instances we need only consider the minimum eigenvalue of the blocks corresponding to accepting witness(es). 
	In these cases, $\gamma_i\leq \eta$ and $\mu_i \geq 1-\eta $, hence
	\begin{align}
	||P_i^{EQMA} - P_i^{QMA}|| &= \max\{\gamma_i^{1/2},(1-\mu_i)^{1/2} \} \\
	&\leq \eta^{1/2}.
	\end{align}
	Thus 
	\begin{align} \label{Eq:YES_Eigenvalue_Bound}
	|\lambda_0 \big(\tilde{H}_{EQMA}^{(YES)}\vert_{\mathcal{K}_S\otimes \mathcal{H}_Q}\big) -\lambda_0\big( H_{QMA}^{(YES)}\vert_{\mathcal{K}_S\otimes \mathcal{H}_Q}\big)| &\leq \eta^{1/2}.
	\end{align}
	Combining the bounds for both the YES and NO cases gives us:
	\begin{equation} \label{Eq:Eigenvalue_Bound}
	|\lambda_0 \big( \tilde{H}_{EQMA}\vert_{\mathcal{K}_S\otimes \mathcal{H}_Q} \big) - \lambda_0 \big( H_{QMA}\vert_{\mathcal{K}_S\otimes \mathcal{H}_Q} \big) | \leq \eta^{1/2}.
	\end{equation}
	
	Now consider the results in the statement of the lemma:
	the YES case is trivial as we know that $\lambda_0 \big(\tilde{H}_{EQMA}^{(YES)}\vert_{\mathcal{K}_S\otimes \mathcal{H}_Q}\big) =0 $ and $H_{QMA}^{(YES)}\vert_{\mathcal{K}_S\otimes \mathcal{H}_Q} \geq 0$, hence using the bound in equation (\ref{Eq:YES_Eigenvalue_Bound}),
	\begin{equation}
	0 \leq \lambda_0\big( H_{QMA}^{(YES)}\vert_{\mathcal{K}_S\otimes \mathcal{H}_Q}\big) \leq \eta^{1/2}.
	\end{equation}
	To see the NO case, we note that the EQMA minimum eigenvalue occurs for a block of the form:
	\begin{align}
	\Delta^{(T)}+ \ket{T}\bra{T}.
	\end{align}
	Using the bound in equation (\ref{Eq:Eigenvalue_Bound}),
	\begin{align}
	1-\cos\bigg(\frac{\pi}{2T}\bigg) - \eta^{1/2} \leq  &\lambda_0\big( H_{QMA}^{(NO)}\vert_{\mathcal{K}_S\otimes \mathcal{H}_Q}\big) \leq 1-\cos\bigg(\frac{\pi}{2T}\bigg).
	\end{align}
\end{proof}

Finally we note that $\eta$ represents the probability of the verifier being wrong. 
If we are interested in a QMA computation, then we can repeat the computation multiple times to get an exponentially better soundness and completeness boundaries \cite{Nielsen_and_Chuang}\cite{Marriott_Watrous_2005}. 
We formalise this below.
\begin{corollary}
	Given a QMA instance, there exists a standard form Hamiltonian, as described in Lemma \ref{Lemma:QMA_Clairvoyance}, which has ground state energy in the bounds
	\begin{align}
	0 \leq &\lambda_0\big( H_{QMA}^{(YES)}\big) \leq e^{-O(\poly(n))} \\
	1-\cos\bigg(\frac{\pi}{2T}\bigg) - e^{-O(\poly(n))} \leq  &\lambda_0\big( H_{QMA}^{(NO)}\big) \leq 1-\cos\bigg(\frac{\pi}{2T}\bigg).
	\end{align}
\end{corollary}
\begin{proof}
	We apply Lemma \ref{Lemma:QMA_Clairvoyance}, where $\eta$ is the probability of the QMA verifier outputting incorrectly: for YES instances $\eta  \geq  \min_{\ket{x}\in \ker H_{in}}\bra{x}\Pi_{out}\ket{x}$, while for NO instances $1 - \eta  \leq  \min_{\ket{x}\in \ker H_{in}}\bra{x}\Pi_{out}\ket{x}$. 
	If we are interested in a QMA computation, then we can repeat the computation a polynomial number of times to get an exponentially better soundness and completeness boundaries \cite{Nielsen_and_Chuang}\cite{Marriott_Watrous_2005}.
	Using these amplification methods, we amplify until $\eta =O(e^{-\poly(n)})$,
	and then apply Lemma \ref{Lemma:QMA_Clairvoyance}, thus giving ground state energies exponentially close to the EQMA case.
\end{proof}

\section{Eigenvalue Scaling with Constant Rejection Probability}

\begin{figure}[h!] 
	\centering
	\includegraphics[width=1.15\textwidth]{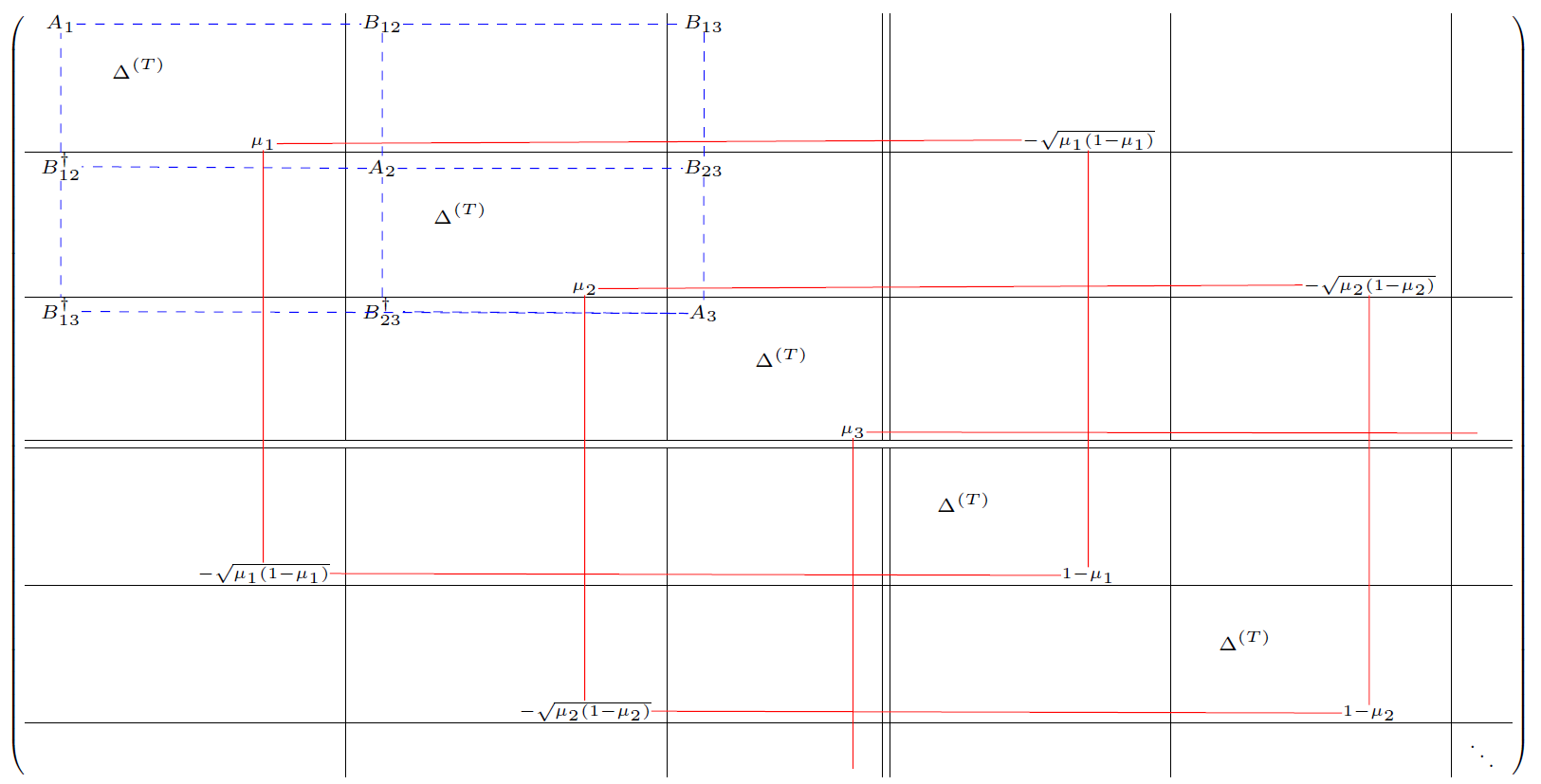}
	\caption{Matrix shows the form of a standard form Hamiltonian $H$ after conjugation by $R$. 
	Each block has size $T\times T$
	Note that $A_i$ and $B_{ij}$ are diagonal $T_{init}\times T_{init}$ sized matrices representing $R^\dagger H_{in}R$.
	The terms connect by solid blue lines are those result from the conjugation of $H_{out}$.
	The solid double-black line separates blocks which are inside and outside of the support of $H_{in}$. 
	}
	\label{Fig:Conjugated_Matrix_Full}
\end{figure}

In the previous section we used the fact that we can amplify in QMA in addition to  perturbation theory to bound the minimum eigenvalues within an exponentially small region. However, it some cases there may be times we cannot amplify, in which case the $\eta^{1/2}$ bound is not useful (since it may be the case $\eta=1/3$). In this section we find an upper bound that scale with $T$ with constant $\eta$.

\begin{theorem}[Constant Acceptance Upper Bound] \label{Theorem:YES_Ground_State_Energy_Bound}
	Let $H^{(YES)}_{QMA}$ be a standard form Hamiltonian with a standard form clock which encodes the verification of a YES QMA instance.
	Let $\eta=O(1)$ be the maximum probability of rejection as per the definition of QMA, then 
	\begin{align}
	0\leq \lambda_0\big( H_{QMA}^{(YES)}\big) \leq \eta\bigg(1-\cos(\frac{\pi}{2(T-T_{init})+1})\bigg)=O\bigg(\frac{\eta}{T^2}\bigg).
	\end{align}
\end{theorem}
\begin{proof}
	From \ref{Lemma:QMA_Clairvoyance}, subspaces of type 1 or 2 have minimum eigenvalue $\lambda_0\big( H_{QMA}^{(YES)}\vert_{\mathcal{K}_S\otimes \mathcal{H}_Q}\big)\geq 1-\cos(\frac{\pi}{O(n)})=\Omega\left(\frac{1}{n^2} \right)$.
	Hence the only subspace where we can hope to get a low upper bound on the minimum eigenvalue is in the legal-only type 3 subspaces.
	For the remainder of this proof we only consider these type 3 subspaces.
	
	We first consider the matrix we are trying to bound the minimum eigenvalue of.
	Denote 
	\begin{align}
	P(\mu_i)= \begin{pmatrix}
	& 1-\mu_i & -\sqrt{\mu_i(1-\mu_i)} \\
	&-\sqrt{\mu_i(1-\mu_i)} & \mu_i
	\end{pmatrix}.
\end{align}
	Let $\mathcal{K}_S\otimes \mathcal{H}_Q$ be a type 3 subspace, then from the Clairvoyance Lemma (Lemma \ref{Lemma:Clairvoyance}) with $R=W(\mathds{1}_C\otimes (X\oplus Y)_Q)$, we have 
	\begin{align}
	A :=& R^\dagger H_{QMA}^{(YES)}\vert_{\mathcal{K}_S\otimes \mathcal{H}_Q}R \\ 
	=& \bigoplus_i(\Delta^{(T)}\oplus \Delta^{(T)}+  0_{T-1} \oplus P(\mu_i)\oplus 0_{T-1}) + R^\dagger  \bigg(\sum_{t=0}^{T_{init}-1}(\ket{t}\bra{t})\otimes  (\Pi_t) \bigg) R. 
	\end{align}
	This has the structure seen in Figure \ref{Fig:Conjugated_Matrix_Full}.
	We now consider an inequality for the minimum eigenvalue using the Rayleigh quotient
	\begin{align} \label{Eq:Restricted_Subspace_Eigenvalues}
	\lambda_0(A) &\leq \min_{\ket{\nu}\in\mathcal{K}_S\otimes \mathcal{H}_Q}\frac{\bra{\nu}A\ket{\nu}}{\bra{\nu}\ket{\nu}} \\
	&\leq \min_{\ket{\nu}\in S \subseteq\mathcal{K}_S\otimes \mathcal{H}_Q }\frac{\bra{\nu}A\ket{\nu}}{\bra{\nu}\ket{\nu}} \label{Eq:Restricted_Subspace_Eigenvalues_2}
	\end{align}
	where $S$ is some restricted subspace and going from the first to second line is a consequence of the min-max theorem for matrix eigenvalues.
	We now consider the structure of $A$. 
	We note that the blocks corresponding to each $\Delta^{(T)}$ which have support on $H_{in}$ are coupled together by the $R^\dagger  \bigg(\sum_{t=0}^{T_{init}-1}(\ket{t}\bra{t})\otimes  (\Pi_t) \bigg) R$ term.
	Then the bottom-right blocks (i.e. those in the complement of the support of $H_{in}$) are disjoint from each other, but each is coupled to a single $\Delta^{(T)}$ block in the top-left by a term $P(\mu_i)$. 
	
	Consider the $P(\mu_i)$ with the smallest value of $\mu_i$ and the two $\Delta^{(T)}$ blocks it couples together. 
	Now restrict this subspace to everything except the first $T_{init}$ rows and columns in the top-left block. 
	This is now completely decoupled from rest of $A$.  
	We label this matrix $B'$ and let the corresponding subspace it acts on be labelled $S$, such that $\dim S=T-T_{init}$. 
	We see
	\begin{align}
	B &:= (\Delta^{(T-T_{init})}+ \frac{1}{2}\ket{0}\bra{0} )\oplus \Delta^{(T)} + 0_{T-T_{init}-1}\oplus P(\mu)\oplus 0_{T-1}.
	\end{align}

	We now consider the inequality above and choose $S$ to be the entire subspace except the states $\{\ket{t}\}_{t=1}^{T_{init}-1}$.
	Then from inequalities (\ref{Eq:Restricted_Subspace_Eigenvalues}) and (\ref{Eq:Restricted_Subspace_Eigenvalues_2}), we have $\lambda_0(A)\leq \lambda_0(B)$. 
	We get the eigenvalue bound
	\begin{align}
	\lambda_0(A) &\leq \min_{\ket{\nu}\in S}\frac{\bra{\nu}B\ket{\nu}}{\bra{\nu}\ket{\nu}} \label{Eq:Variational_Eigenvalue_1}.
	\end{align} 
	
	From now on denote $T':=T-T_{init}$.
	Further let the minimum eigenvector of $\Delta^{(T')}+ \ket{0}\bra{0}$ be $\ket{u}$, where by Theorem 3.2 $(viii)$ of \cite{Yueh_Cheng_2008} its components are given by
	\begin{align}
	\ket{u} = \sum_{t=1}^{T'} u_t\ket{t} = \sqrt{\frac{2\sin(\pi/(4T'+2))}{\sin(\pi T'/(2T'+1))}}  \sum_{t=1}^{T'} \sin(\frac{t\pi}{2T'+1})\ket{t-1},
	\end{align}
	and the associated eigenvalue is $\gamma_0 = 1 - \cos(\pi/(2T'+1))$.
	Furthermore, $P(\mu)$ has an eigenvector with zero eigenvalue
	\begin{align}
	\sqrt{\mu}\ket{0}+ \sqrt{1-\mu}\ket{1},
	\end{align}
	and $\Delta^{(T)}$ has an eigenvector with zero eigenvalue
	\begin{align}
	\ket{w} = \frac{1}{\sqrt{T}}\sum_{t=1}^{T}\ket{t}.
	\end{align}
	We then use the unnormalised vector 
	\begin{align}
	\ket{\nu} = \begin{pmatrix}
	\frac{1}{u_{T'}}\sqrt{\mu}&\ket{u} \\
	\sqrt{T}\sqrt{1-\mu}&\ket{w}
	\end{pmatrix}
	\end{align}
	as a trial ground state.
	Consider 
	\begin{align}
	\frac{\bra{\nu}B\ket{\nu}}{\bra{\nu}\ket{\nu}} &= \frac{(\mu/u_{T'}^2) \gamma_0}{(\mu/u_{T'}^2)+ (1-\mu)T} \\
	 &= \frac{\mu \gamma_0}{\mu+ (1-\mu)Tu_{T'}^2}.
	\end{align}
	We now note that 
	\begin{align}
	T|u_{T'}|^2 &= \frac{2T\sin(\pi/(4T'+2))}{\sin(\pi T'/(2T'+1))}\sin[2](\frac{T'\pi}{2T'+1}) \\
	&= 2T\sin(\frac{\pi}{4T'+2})\sin(\frac{T'\pi}{2T'+1}) \\
	&= 2T\sin(\frac{\pi}{4T'+2})\cos(\frac{\pi}{4T'+2}) \\
	&=T\sin(\frac{\pi}{2T'+1})
	\end{align}
	From the Clock Properties Assumptions (Assumptions \ref{Assumption:Clock_Properties}) we know that $T_{init}=O(T^{1/2})$, hence for  $T\geq 2$ we find that 
	\begin{align}
	T|u_{T'}|^2 &\geq 1.
	\end{align}
	Using this we get
	\begin{align}
	\frac{\bra{\nu}B\ket{\nu}}{\bra{\nu}\ket{\nu}} &\leq \mu \gamma_0.
	\end{align}
	We know that $\gamma_0 = 1 - \cos(\pi/2(T-T_{init})+1)$, hence, combining all of the above, we get 
	\begin{align}
	\frac{\bra{\nu}H^{(YES)}_{QMA}\ket{\nu}}{\bra{\nu}\ket{\nu}} &\leq \mu \bigg(1-\cos(\frac{\pi}{2(T-T_{init})+1})\bigg).
	\end{align}
	Again using $T_{init}=O(T^{1/2})$ we get
	\begin{align}
	\lambda_0(H^{(YES)}_{QMA}) = O\bigg(\frac{\mu}{T^2}\bigg).
	\end{align}
	Finally, we note $\mu \leq \eta$ as $\eta$ is the maximum probability of rejection of a correctly initialised witness, thus giving
	\begin{align}
	\lambda_0(H^{(YES)}_{QMA}) = O\bigg(\frac{\eta}{T^2}\bigg).
	\end{align}
	
\end{proof}

\section{Discussion and Outlook}

The main aim of this work has been to understand the ground state eigenvalues and subspace as thoroughly as possible, as  well as providing a toolbox for future work involving Feynman-Kitaev Hamiltonians and their extensions. 
We expect the constant rejection probability analysis to be useful for situations where our ability to provide amplifications is limited in some way. 
An example is in \cite{Bausch_Cubitt_Watson_2019} where we are allowed only very limited amplification and we need to apply these bounds.

We can then ask what else can we apply this analysis to:
\paragraph{Extensions to Unitary Labelled Graphs}
As mentioned previously, unitary labelled graphs representing branching computations have been shown to have promise gaps going as $\Omega(N^{-3})$ if the graph has $N$ vertices \cite{Bausch_Cubitt_Ozols_2016}. 
Other bounds are known, but are similarly fairly loose. 
Given some of the techniques in this paper (notably the Uncoupling Lemma (Lemma \ref{Lemma:Divorce_Lemma}) allow us to decouple line graphs, it would be interesting to see if better bounds for ULGs can be found using these techniques or something similar.

\paragraph{Analysis of Quantum Walks on a Line}
The analysis quantum walks on a line in this paper is limited in that it only tells us that energy penalties at the end points are lower energy than elsewhere, and then given bounds on the energy of these end points. 
Intuitively, one would expect energy penalties closer to the centre of a computation to raise the energy more.	
It would be useful if we could rigorously understand how placing an energy penalty within a computation can change the energy --- indeed such results would liked help us proving better bounds for ground state energies of unitary labelled Hamiltonians.

\paragraph{Fine-tuned Hamiltonian Energies}
A wider project can be seen in the context of designing Hamiltonians with specifically chosen energies and associated scalings for use in particular constructions. 
Examples include \cite{Bausch_Cubitt_Watson_2019} and \cite{Bausch_Cubitt_Lucia_PG_2018}, both of which rely on a construction where the energy of a negative energy Hamiltonian is traded-off against a positive energy Hamiltonian encoding a computation.

There are two motivating points in this: extending the quantum walk analysis to bonus penalties to get a Hamiltonian which as an energy $-f(T)$, for runtime $T$, such that $f(T)$ only decays polynomially. 
The second point would be attempting to encode a computation in a Hamiltonian with negative ground state energies.
At the moment it is not known how to do this due to difficulties initialising the encoded circuit/quantum Turing Machine. 
That is, the bonus provided by reaching an accepting state is usually sufficiently large to make it favourable to pick up energy penalties in the initialisation steps.

\section{Acknowledgements}
The author would like to thank Toby Cubitt for support and useful discussions particularly regarding Section \ref{Sec:Standard_Form_Hamiltonian_Analysis}, and Johannes Bausch for discussions regarding the constant acceptance probability lemma.
The author is supported by the EPSRC Centre for Doctoral Training in Delivering Quantum Technologies.

\end{document}